\documentclass{article}

\usepackage{arxiv}
\usepackage{natbib}
\usepackage[utf8]{inputenc} 
\usepackage[T1]{fontenc}    
\usepackage{url}            
\usepackage{booktabs}       
\usepackage{amsfonts}       
\usepackage{nicefrac}       
\usepackage{microtype}      
\usepackage{lipsum}         
\usepackage{graphicx}
\usepackage{doi}
\usepackage[inline]{enumitem}
\usepackage{amsmath}
\usepackage{mathtools} 
\usepackage{amsthm}         

\DeclareMathOperator*{\argmax}{arg\,max}

\usepackage{subcaption}     
\usepackage{algorithm}
\usepackage{algorithmic}

\usepackage{hyperref}       
\usepackage{cleveref}       

\title{Robust Self-Triggered Control Approaches Optimizing Sensors Utilization  with Asynchronous Measurements}

\newif\ifuniqueAffiliation
\uniqueAffiliationtrue

\ifuniqueAffiliation 
\author{ \href{https://orcid.org/0000-0001-5024-9224}{\includegraphics[scale=0.06]{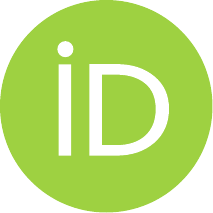}\hspace{1mm}Abbas Tariverdi}\thanks{This work was conducted while Abbas Tariverdi was with CRIStAL (Research Center in Computer Science, Signal and Automatic Control of Lille), UMR 9189, CNRS, University of Lille, Centrale Lille, Lille, France.} \\
	\texttt{abbasta@abbasta.com} 
}
\else
\usepackage{authblk}

\newbox{\orcid}\sbox{\orcid}{\includegraphics[scale=0.06]{orcid.pdf}} 
\author[1]{%
	\href{https://orcid.org/0000-0000-0000-0000}{\usebox{\orcid}\hspace{1mm}David S.~Hippocampus\thanks{\texttt{hippo@cs.cranberry-lemon.edu}}}%
}
\author[1,2]{%
	\href{https://orcid.org/0000-0000-0000-0000}{\usebox{\orcid}\hspace{1mm}Elias D.~Striatum\thanks{\texttt{stariate@ee.mount-sheikh.edu}}}%
}
\affil[1]{Department of Computer Science, Cranberry-Lemon University, Pittsburgh, PA 15213}
\affil[2]{Department of Electrical Engineering, Mount-Sheikh University, Santa Narimana, Levand}
\fi

\renewcommand{\shorttitle}{Self-Triggered Controls Optimizing Sampling Sequences with Asynchronous Measurements}

\newtheorem{proposition}{Proposition}

\newtheorem{remark}{Remark} 
\newtheorem{assumption}{Assumption}    
\newtheorem{definition}{Definition}    

\hypersetup{
pdftitle={A template for the arxiv style},
pdfsubject={q-bio.NC, q-bio.QM},
pdfauthor={David S.~Hippocampus, Elias D.~Striatum},
pdfkeywords={First keyword, Second keyword, More},
}

\begin{document}
\maketitle



\begin{abstract}
\textit{Note: This research was conducted in 2017--2018. The literature review has not been updated and may not reflect subsequent or concurrent developments in the field.} 

Most control systems run on digital hardware with limited communication resources. This work develops self-triggered control for linear systems where sensors update independently (asynchronous measurements). The controller computes an optimal horizon at each sampling instant, selecting which sensor to read over the next several time steps to maximize inter-sample intervals while maintaining stability.

Two implementations address computational complexity. The online version solves an optimization problem at each update for theoretical optimality. The offline version precomputes optimal horizons using conic partitioning, reducing online computation to a lookup. Both guarantee exponential stability for unperturbed systems and global uniform ultimate boundedness for systems with bounded disturbances. Simulations demonstrate 59-74\% reductions in sensor utilization compared to periodic sampling. The framework enables resource-efficient control in networked systems with communication constraints.
\end{abstract}

\keywords{Event-triggered control \and global uniform ultimate boundedness \and sampled-data systems \and self-Triggered control \and network control systems}

\section{Introduction}\label{Intro}

Most control systems now run on digital hardware. Networked Control Systems (NCSs) require careful management of limited computational and communication resources \cite{heemels2012introduction}. This drives fundamental design questions: when to sample sensors, when to update control actions, and when to transmit data \cite{miskowicz2018event}.

Event-based control addresses these resource constraints by updating control inputs only when necessary. Two approaches exist: \textit{event-triggered control}, which monitors measurements continuously and triggers updates when a threshold condition is violated, and \textit{self-triggered control}, which precomputes the next update time using sampled measurements and plant dynamics \cite{heemels2012introduction}.

The key challenge is stability. Self-triggered mechanisms must enlarge average sampling intervals, reducing computational and energy costs, while maintaining system stability. This work designs self-triggered state-feedback controllers for linear systems that achieve this balance. Self-triggered control for systems with \textit{asynchronous measurements} presents unique challenges. Unlike synchronous sampling where all sensors update simultaneously, asynchronous measurements arrive at different rates. This is common when sensors have varying priorities, power constraints, or communication delays.

Stability analysis for aperiodic sampled-data systems has developed along four main approaches \cite{hetel2017recent}: time-delay modeling, hybrid systems, discrete-time methods, and input-output stability. The \textit{time-delay approach} \cite{mikheev1988asymptotic,fridman2003delay} models the system as continuous-time with time-varying delay, extending naturally to asynchronous measurements where each sensor introduces an independent delay. The \textit{hybrid systems approach} \cite{dullerud1999asynchronous,toivonen1992sampled} models sampling as discrete jumps in continuous dynamics; here, each asynchronous sensor creates its own jump sequence.

The \textit{discrete-time approach} analyzes the spectral radius or constructs appropriate Lyapunov functions. Methods include quasi-quadratic \cite{molchanov1989criteria}, parameter-dependent \cite{daafouz2001parameter}, and path-dependent \cite{lee2006uniform} formulations. For asynchronous systems, this approach addresses packet dropouts and time-varying delays that exceed sampling intervals \cite{cloosterman2010controller,van2010tracking}, common when sensors operate independently. Finally, the \textit{input-output approach} treats sampling and delay as perturbations to the nominal control loop, using small-gain conditions to quantify how independent sensor delays affect closed-loop stability margins \cite{mirkin2007some,fujioka2009stability}.

Self-triggered control leverages these foundations to precompute the next update time using current measurements and plant dynamics. No continuous monitoring is required. Early work discretized plants for linear systems \cite{velasco2003self} or designed $H_\infty$ controllers balancing CPU utilization with performance \cite{wang2008state}. Extensions to nonlinear systems followed, building on event-triggering foundations \cite{tabuada2007event,anta2010sample}.

A key advance came from ISS-based approaches \cite{mazo2009self,mazo2010iss}. Using Lyapunov-based triggering conditions, these methods guarantee exponential input-to-state stability while enlarging sampling intervals. Implementation uses discrete-time evaluation: periodic checking rather than continuous monitoring. Wang and Lemmon extended this framework to guarantee finite-gain $L_2$ stability through continuous Lyapunov decay \cite{wang2009self,wang2010self}, handling arbitrary disturbances independent of the process model.

The relevance of this approach is clear: precomputed sampling decisions dramatically reduce computational overhead compared to continuous event monitoring. The question for asynchronous measurements is whether this advantage persists when sensors operate on independent schedules rather than synchronized clocks.

This paper develops optimal self-triggering mechanisms that maximize average sampling intervals while guaranteeing exponential stability for both unperturbed and perturbed linear systems. At each sampling instant, the controller computes an \textit{optimal horizon}, a sequence specifying which sensor (or none) to read over the next several time steps. This transforms the problem from simple triggering to active sensor scheduling.

A major complectity in optimal self-triggered control is computational intensity. Solving an optimization problem at every sampling instant is often impractical for embedded real-time implementation. We address this through two distinct implementations. The \textit{online implementation} solves the optimization problem at each step, demonstrating theoretical optimality but demanding higher computation. The \textit{offline implementation} utilizes conic partitioning to precompute optimal horizons for specific state-space regions. This reduces the online execution to a low-complexity lookup table, making it suitable for real-time deployment.

Section \ref{Des} establishes the system model and problem formulation. Section \ref{Unperturbed} analyzes ideal conditions, deriving both implementations with stability proofs. Section \ref{Perturbed} extends the framework to systems with bounded disturbances. Section \ref{Simulation} validates the approach through numerical studies, and Section \ref{Conclusion} summarizes the contributions. Detailed proofs are provided in Section \ref{Appendix}.

\section{System Description and Problem Statement}\label{Des}

\subsection{System Description}
We consider a continuous time LTI control system
\begin{equation} \label{syst}
	\begin{split}
		\dot{x}(t)&=Ax(t)+Bu(t),~\forall t \geq 0, 
	\end{split}
\end{equation}
where $ x \in \mathbb{R}^n $ is the system's state, $ u \in \mathbb{R}^m $ is the system's input and $ A $ and $ B $ are matrices of appropriate dimensions. We consider a sampled-data controller
\begin{equation*} \label{fc}
	u(t)=u^*(t_h),~ \forall t_h=hT,~h\in \mathbb{N},
\end{equation*}
in which $ u^*(t_h) $ is constant $ \forall t \in [t_h,t_h+T) $, and the sampling instants $ t_h,~h \in \mathbb{N} $, satisfy
\begin{equation} \label{samplint}
	T=t_{h+1}-t_h>0, ~\forall h\in \mathbb{N},~\text{with} ~ t_0=0.
\end{equation}

The discrete-time model of the system \eqref{syst}  at instants $ t_h $ can be written as 
\begin{equation} \label{dsyst}
	\begin{split}
		{x}(t_{h+1})=e^{AT}x(t_h)+\int\nolimits_{0}^{T}e^{As}B u^*(t_h) ds 
	\end{split}.
\end{equation}

The aim of the paper is to propose a self triggering mechanism in which the next sampling instant $ t_{h+1},~ h\in\mathbb{N} $, is calculated at each sampling instant $ t_h,~h\in \mathbb{N} $, based on partial and/or full information about $ x(t_h) $ and an optimization criterion.

\subsection{System Reformulation}\label{sys-refor}

In classical work concerning self triggering control, at each sampling instant, the next optimal (based on some certain criteria) admissible sampling intervals are computed. In this work, we want to compute the next optimal sampling intervals over a finite horizon. Therefore, we need to introduce remodeling of the system and a few notations regarding sampling horizons and sampling sequences.

We consider a scenario in which all sensors may not be used simultaneously at each sampling instant $ t_h $. In other words, at each sampling instant $ t_h $, we suppose that the proposed mechanism just may use one sensor's measurement while preserves the stability of the system. Therefore, we consider a sampled-data sate-feedback controller
\begin{equation} \label{statefc}
	u^*(t_h)=K\hat{x}(t_h),~ \forall t_h=hT,~h\in \mathbb{N},
\end{equation}
in which $ K $ is a given feedback gain matrix such that the sampling instants $ t_h,~h \in \mathbb{N} $, satisfy
\begin{equation*} \label{samplints}
	T=t_{h+1}-t_h>0, ~\forall h\in \mathbb{N},~\text{with} ~ t_0=0.
\end{equation*}
and
\begin{equation*}
	\begin{split}
		\hat{x}(t_h)=\begin{bmatrix}
			\hat{x}^{(1)}(t_h),
			\cdots,
			\hat{x}^{(m)}(t_h)
		\end{bmatrix}^T\in\mathbb{R}^n, ~ \hat{x}(t_0)\in\mathbb{R}^n
	\end{split},
\end{equation*}
where,
\begin{equation}\label{CntrlMS}
	\begin{split}
		\hat{x}^{(i)}(t_h)=    \left\{\begin{array}{ll}
			x^{(i)}(t_h) & \text{If sensor $ i^{th} $ is sampled at $ t_h $}. \\
			\hat{x}^{(i)}(t_{h-1}) & \text{Otherwise}.
		\end{array}\right.
	\end{split}
\end{equation} 
A description of the system with asynchronous measurement is depicted in Fig.~{ ref{SystemDes}}. 
\begin{figure}[!htbp] 
	\centering
	\subfloat{%
		\resizebox*{0.5\textwidth}{!}{\includegraphics{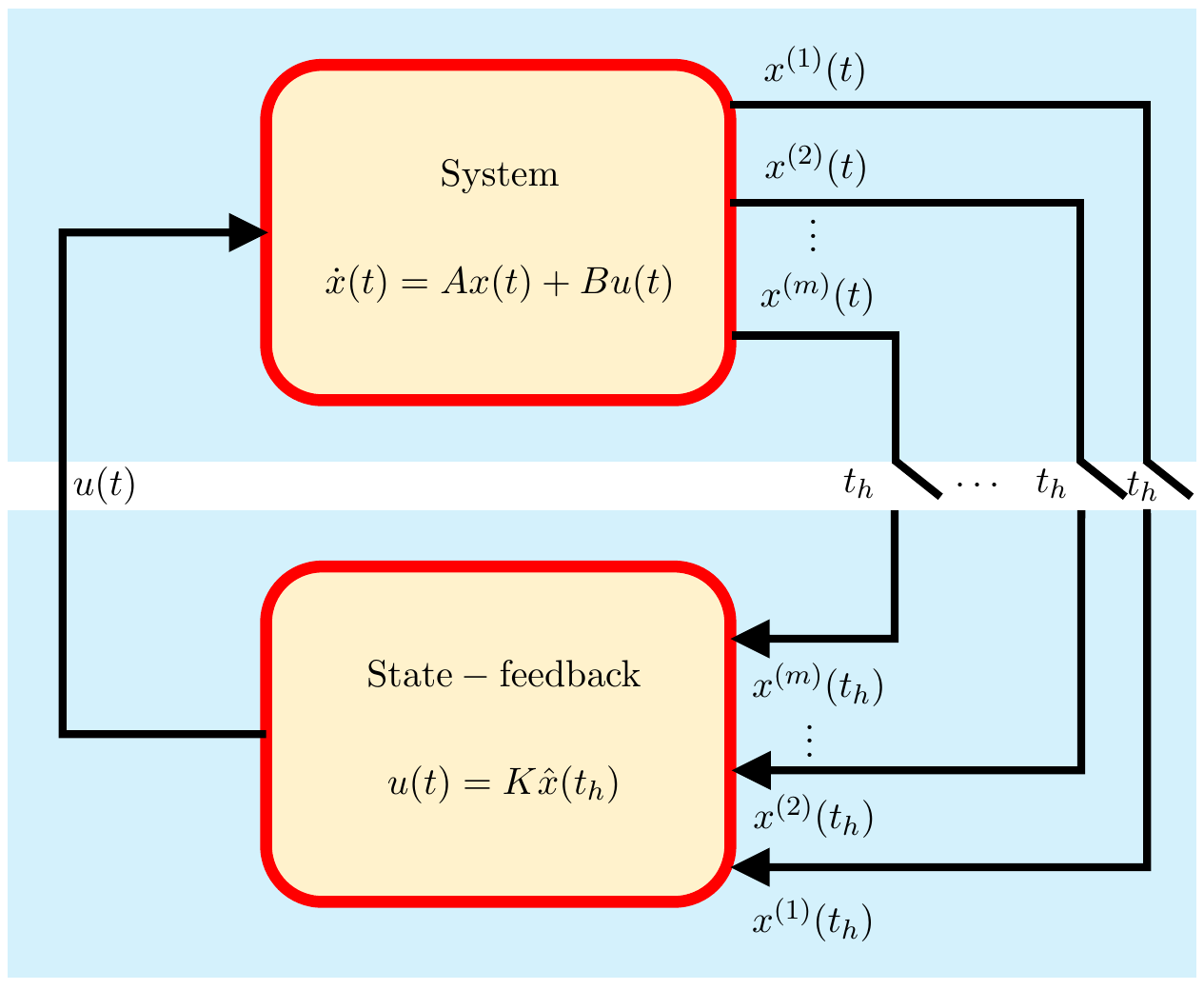}}}\hspace{2pt}
	\caption{ System description with asynchronous measurement.} 	\label{SystemDes}
\end{figure} 

Equivalently, \eqref{CntrlMS} can be written as
\begin{equation}\label{deltalMS}
	\hat{x}^{(i)}(t_h)= \delta^{(i)}(t_h)x^{(i)}(t_h)+(1-\delta^{(i)}(t_h))\hat{x}^{(i)}(t_{h-1}),
\end{equation}  
with
\begin{equation*}
	\begin{split}
		\delta^{(i)}(t_h)=    \left\{\begin{array}{ll}
			1 & \text{If sensor $ i^{th} $ is sampled at $ t_h $}. \\
			0 & \text{Otherwise}.
		\end{array}\right.
	\end{split}
\end{equation*}

In collective form, \eqref{deltalMS} can shown as
\begin{equation}\label{cntrlMSC}
	\begin{split}
		\hat{x}(t_h)=M(t_h)x(t_h)+N(t_h)\hat{x}(t_{h-1})
	\end{split},
\end{equation}
with
\begin{equation*}
	\begin{split}
	M(t_h) & =\text{diag}\{\delta^{(1)}(t_h)I_{n_1}, \delta^{(2)}(t_h)I_{n_2}, \cdots, \\
	& \quad \delta^{({m-1})}(t_h)I_{n_{m-1}},\delta^{(m)}(t_h)I_{n_m}\}.  
\end{split}  
\end{equation*}
and
\begin{equation*}
	\begin{split}
	N(t_h) & =\text{diag}\{(1-\delta^{(1)}(t_h))I_{n_1}, (1-\delta^{(2)}(t_h))I_{n_2},\cdots,\\
	& \quad (1-\delta^{(m)}(t_h))I_{n_m} \}.
\end{split} 
\end{equation*}  

From \eqref{statefc} and \eqref{cntrlMSC}, the discrete-time model of the system \eqref{dsyst}  at instants $ t_h $ can be written as
\begin{equation}\label{MS}
	\begin{split}
	{x}(t_{h+1}) & =e^{AT} x(t_h)+\int\nolimits_{0}^{T}e^{As}Bds K \hat{x}(t_h),\\
	\hat{x}(t_h) & =M(t_h)x(t_h)+N(t_h)\hat{x}(t_{h-1}).
\end{split} 
\end{equation} 
in the collective form, equation \eqref{MS} yields
\begin{equation}\label{SysC}
	\eta_{h+1}=\tilde{A}_{(\mathtt{S}_h)}\eta_{h},~ \forall t_h=hT,~h\in \mathbb{N},~\eta_0=[x(t_0) ~ x(t_0)]^T
\end{equation}
with
\begin{equation*}
	\eta_h=[x(t_h) ~ \hat{x}(t_{h-1})]^T,~\eta_{h+1}=[x(t_{h+1}) ~ \hat{x}(t_{h})]^T,
\end{equation*} 
and
\begin{equation}\label{discollective}
	\tilde{A}_{(\mathtt{S}_h)}= \begin{bmatrix}
		e^{AT}+\int\nolimits_{0}^{T}e^{As}B Kds M(t_h) & \int\nolimits_{0}^{T}e^{As}B Kds N(t_h)\\
		M(t_h) & N(t_h)
	\end{bmatrix}.
\end{equation}
%
%
	
	$\sigma=\{\mathtt{S}_{\sigma}^j\}_{j=1}^{l}=(\mathtt{S}_{\sigma}^1 \mathtt{S}_{\sigma}^2\cdots\mathtt{S}_{\sigma}^l) $, with $l \in \mathbb{N} \setminus \{0\} $, refers to a horizon of sampled sensors with length $ l $, where $ j $ represents the position of a sampled sensor $ \mathtt{S}_{\sigma}^j$ inside the horizon. For a finite set $ \mathcal{S} \subset \mathbb{N} $, we define $ S_{l_{min}}^{l_{max}}(\mathcal{S}) $ the set of all horizons $ \sigma=\{\mathtt{S}_{\sigma}^j\}_{j=1}^{l} $ of length $ l \in [l_{min},l_{max}] $ with values $ \mathtt{S}_{\sigma}^j \in \mathcal{S},~ \forall j\in \{1,\cdots,l\} $:
	\begin{equation}\label{key1}
\begin{split}
			S_{l_{min}}^{l_{max}}(\mathcal{S}) &=\Big \{ \sigma=\{\mathtt{S}_{\sigma}^j\}_{j=1}^{l}: l \in [l_{min},l_{max}], \text{and}, \\
			\mathtt{S}_{\sigma}^j &\in \mathcal{S} ,~\forall j\in \{1,\cdots,l\} \Big \}.
\end{split}
	\end{equation}
	
	We denote by $ \{\sigma_k\}_{k\in \mathbb{N}}$ a sampling sequence composed by sampling horizons $ \sigma_k=(\mathtt{S}_{\sigma_k}^1\cdots \mathtt{S}_{\sigma_k}^{l_k}) \in S_{l_{min}}^{l_{max}}(\mathcal{S}),~ k \in \mathbb{N}  $. The scalars $ \mathtt{S}_{\sigma_k}^{i}, ~ k \in \mathbb{N},~ i\in \{1,\cdots,l_k\} $ define sensors which are used and constitute the sampling sequence where
	$ k $ indicates the index of the horizon and $ i $ the position of this sampling step in the considered horizon $ \sigma_k $. 
	
	The representation of system \eqref{SysC} over a sampling sequence $\{\sigma_k\}_{k\in \mathbb{N}} $ is given as
	\begin{equation} \label{Hdsyst}
		\begin{split}
			{x}{_{k+1}}=\Phi_{\sigma_k} x_k ,~ k\in \mathbb{N},
		\end{split}
	\end{equation}
	with
	\begin{equation} \label{Hdsyst1}
		\begin{split}
			\Phi_{\sigma_k}={\tilde{A}}_{(\mathtt{S}^{l_k}_{\sigma_k})}{\tilde{A}}_{(\mathtt{S}^{l_{k-1}}_{\sigma_k})}\cdots {\tilde{A}}_{(\mathtt{S}^{l_{1}}_{\sigma_k})},~x_k=x(\tau_k).
		\end{split}
	\end{equation}
	The instant $ \tau_k,~ {k\in \mathbb{N}} $ denotes the starting time of a sampling sequence $ \sigma_k $ such that
	\begin{equation}
		\tau_{k+1}=\tau_{k}+l_k T,~ \tau_0=t_0=0  
	\end{equation} 
	
	A description of these sequences is depicted in Fig.~{\ref{InstantA}}. 
	\begin{figure}[!htbp]
		\centering
		\subfloat{%
			\resizebox*{8cm}{!}{\includegraphics{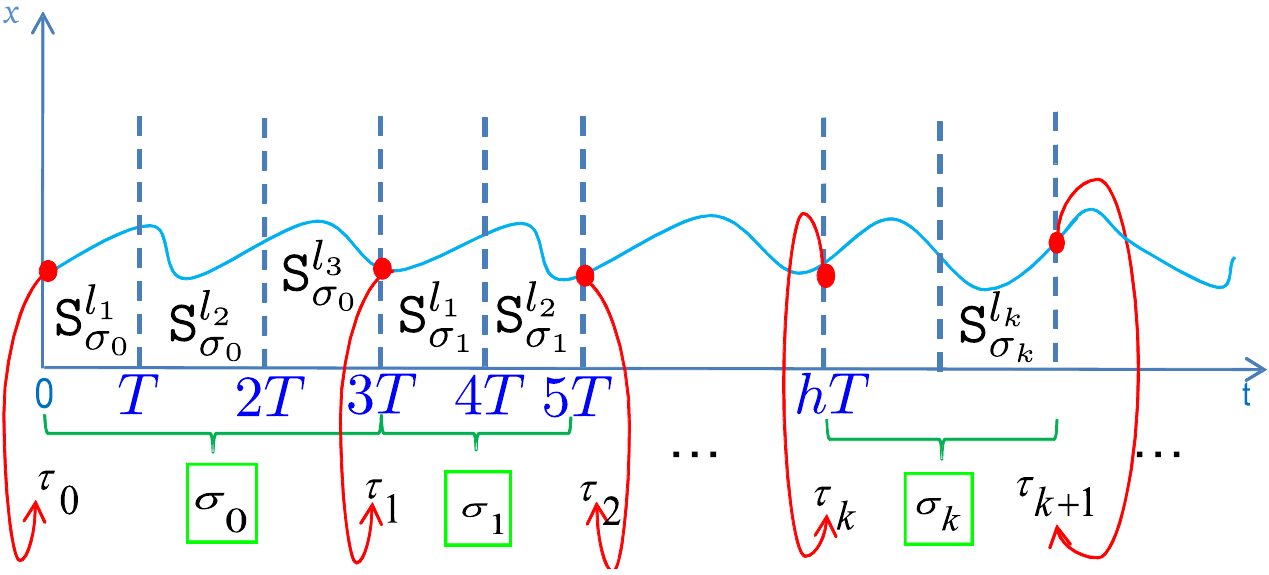}}}\hspace{2pt}
		\caption{ System state discretization using asynchronous measurement.} 	\label{InstantA}
	\end{figure}
	
\subsection{Motivational Example}

Consider a system of the form \eqref{dsyst} with a constant sampling interval $ T_h=T,~ \forall h \in \mathbb{N} $, and the following characteristics
\begin{equation}\label{exp1}
	A=\begin{bmatrix} 0 & 1\\-2 & 3 \end{bmatrix},~ B=\begin{bmatrix} 0 \\1 \end{bmatrix}, ~ K = \begin{bmatrix} 1 & -4 \end{bmatrix}.
\end{equation}

To clarify our motivation, let us look at the transition matrix $ {\tilde{A}}_{(2)}^4 \times {\tilde{A}}_{(1)}^1  \times  {\tilde{A}}_{(0)}^2$. According to the notation, $ {\tilde{A}}_{(i)}$  refers to the discretized collective form \eqref{discollective} which is obtained by using the sensor $ i^{th} $'s information in each constant period. 

Fig.~{\ref{MeigSys}} shows that the system stability and the average of sampling intervals for each sensor depend on the sampling interval $ T $ and a sensor is used in each sampling period. For the transition matrix $ {\tilde{A}}_{(2)}^4 \times {\tilde{A}}_{(1)}^1  \times  {\tilde{A}}_{(0)}^2$ and the constant sampling period $ T= 0.297s $, the average of sampling intervals for each sensor is $ 0.9334 $ which is considerably greater than $ T_{Schur}=0.59 $.  $ T_{Schur} $ is the maximum allowable periodic sampling interval that stabilize the system provided that all sensors' information is used at each sampling period.
\begin{figure}[!htpb] 
	\centering
	\subfloat{%
		\resizebox*{10cm}{!}{\includegraphics{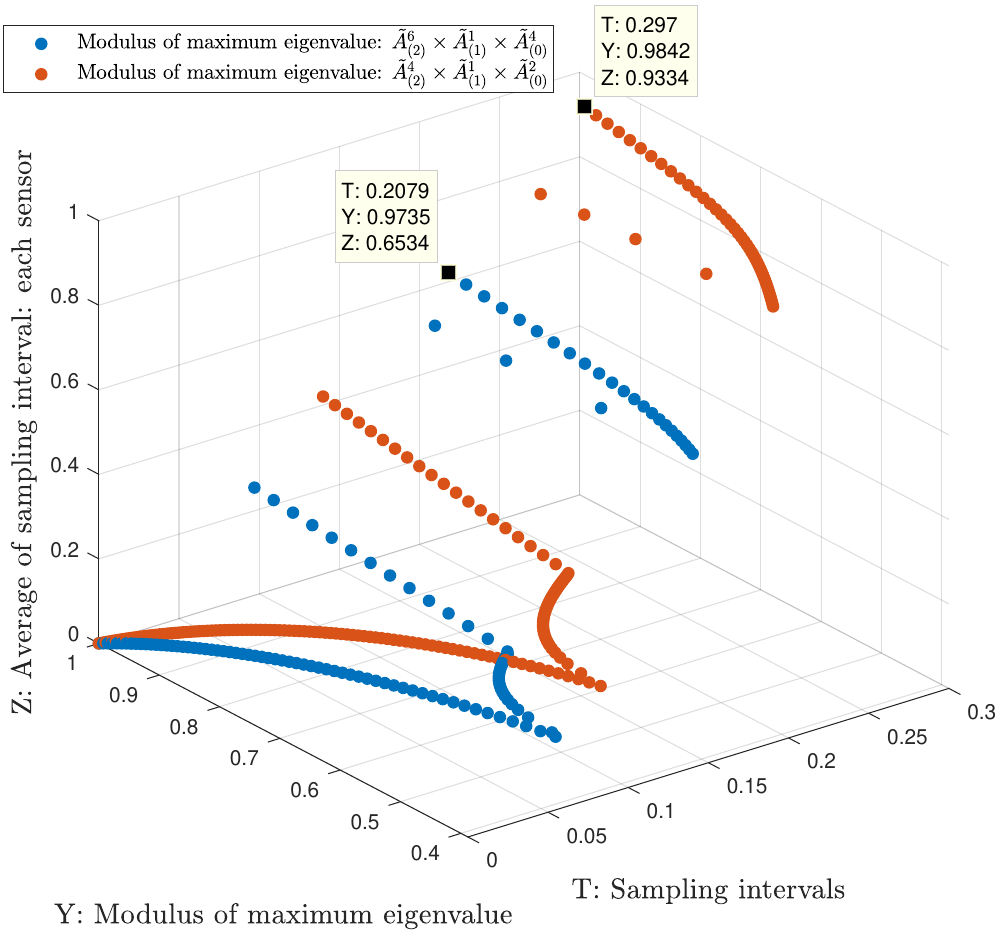}}}\hspace{2pt}
	\caption{ Evolution of the average of sampling intervals for each sensor with respect to modulus of the maximum eigenvalue and stabilizing sampling intervals  } 	\label{MeigSys}
\end{figure}

Motivated by this observation, this work aims at characterizing other sampling sequences that not only stabilize the system \eqref{dsyst} but also have a maximum average of sampling intervals for each sensor among other possible sequences based on the current state of the system. 

\subsection{Problem Statement}

Our aim is to calculate, for a current state $ x_k=x(\tau_k) $ at instant $ \tau_k $, the next sampling horizon $ \sigma_k \in S_{l_{min}}^{l_{max}}(\mathcal{S}) $ for a given $ l_{min}$ and  $l_{max} $ and $ \mathcal{S} $ which has a maximum possible average of sampling intervals for each sensor and  ensures the system's exponential stability.

\section{Main Results: Perturbation-free Case}\label{Unperturbed}
\subsection{Proposed Self Triggering Mechanism and Stability Analysis: Online case}\label{STM-off}
\subsubsection{Self Triggering Mechanism and Stability Analysis}
Proposition \ref{S_O_Df} ensures exponential stability of the discretized system \eqref{Hdsyst} in which the next sampling horizon $ \sigma_k $ defined by proposed self-triggering mechanism.
\begin{proposition}\label{S_O_Df}
	Consider a scalar $ \beta>0 $, a constant sampling interval $ T $, and a matrix $ P=P^T\succ0 $ such that
	\begin{equation} \label{matrixp}
		\begin{split}
			&\exists \sigma^* \in S^{l_{max}}_{l_{min}}(\mathcal{S}): \Phi^T_{\sigma^*} P \Phi_{\sigma^*} -e^{(-\beta |\sigma^* |T) }P\prec0
		\end{split} .
	\end{equation}
	
	Then, system \eqref{Hdsyst} with the self-triggering mechanism defined by
	
	\begin{equation} \label{STM-on-wp}
		\begin{split}
			& \sigma_k  \in \argmax_{{\sigma} \in {{\bar{S}}^{l_{max}}_{l_{min}}}(\mathcal{S},x_k)}\Bigg\{\frac{\sum_{j\in\{\mathtt{S}_{\sigma}^j = 0\}}1+|\sigma | }{|\mathcal{S}||\sigma |} \Bigg\},
		\end{split}
	\end{equation}	
	with
	\begin{equation}\label{es-wp}
		{{\bar{S}}^{l_{max}}_{l_{min}}}(\mathcal{S},x_k) =\Bigg\{\sigma \in S^{l_{max}}_{l_{min}}(\mathcal{S}): \zeta_k \leq 0 \Bigg\},
	\end{equation}
	is exponentially stable with the decay rate of $ \beta/2 $, where $ \zeta_k =  x^T_k \Big(\Phi^T_{\sigma} P \Phi_{\sigma} -e^{(-\beta |\sigma |T) }P  \Big)x_k$.
\end{proposition}
\begin{proof}
	Consider a quadratic Lyapunov function $ V:\mathbb{R}^n  \to \mathbb{R}^+ $ defined by $ V(x_k)=x^T_k P x_k $ in which the matrix $ P $ satisfies \eqref{matrixp}. Then by construction, the set $ {{\bar{S}}^{l_{max}}_{l_{min}}}(\mathcal{S},x_k) $ defined by \eqref{es-wp} is not empty (i.e. $ \sigma^* \in {\bar{S}}^{l_{max}}_{l_{min}}(\mathcal{S},x_k),~ \forall x_k \in \mathbb{R}^n $). Therefore,  the arguments of the maxima in the equation \eqref{STM-on-wp} is well-defined and non-empty. Then $ \forall k \in \mathbb{N} $, by the equation \eqref{STM-on-wp} we can conclude that
	\begin{equation}
		\begin{split}
			V(x_{k+1})&= x_{k+1}^T P x_k = x^T_k \Phi^T_{{\sigma}_k} P\Phi_{{\sigma}_k} x_k \\
			& \leq  e^{(-\beta |\sigma_k |T) }x^T_k P x_k = e^{-\beta (\tau_{k+1}-\tau_k)} V(x_k).
		\end{split}
	\end{equation}
	Therefore, $ \forall k \in \mathbb{N} $, by recursion, we have  $ V(x_k) \leq e^{-\beta \tau_k} V(x_0) $, and according to \cite{khalil2002nonlinear}, system \eqref{Hdsyst} with self triggering control \eqref{STM-on-wp} is exponentially stable with the decay rate of $ \beta/2 $.
\end{proof}

\begin{remark}
	Self-triggering mechanism defined by \eqref{STM-on-wp} chooses a next optimal sequence from the set \eqref{es-wp} which contains all stable sequences such that the chosen sequence has a maximum average of sampling intervals among others.
\end{remark}

\subsubsection{Implementation Algorithm }
Algorithm \ref{Alg1} provides a self triggering mechanism for the sampled-data system \eqref{Hdsyst} which generates the next stabilizing and optimal sampling horizon $ {\sigma}_k $ based on the sampled state  $ x_k=x(\tau_k) $.

\begin{algorithm}
	\caption{Self Triggering Mechanism for Sampled-Data System \eqref{Hdsyst}: Online case}\label{Alg1}
	\begin{algorithmic}[H]
		\STATE
		\STATE \textbf{{OFFLINE Procedure:}}
		\begin{itemize}
			\STATE Define $ l_{min} $ and $ l_{max}  $, the minimum and maximum length of the considered sampling horizons.
			\STATE Define a constant sampling interval $ T $ and  a finite set $ \mathcal{S}=\{\mathtt{S}^1, \mathtt{S}^2,\cdots,\mathtt{S}^m\} $ of sensors' number.
			\STATE Design the set $ S^{l_{max}}_{l_{min}}(\mathcal{S}) $ of sampling horizons of length $l \in 	[l_{min},l_{max}] $ constituted of sensors in $ \mathcal{S} $.
			\STATE Compute the matrix $ P=P^T\succ0 $ for a given stable sequence $ \sigma^*\in S^{l_{max}}_{l_{min}}(\mathcal{S}) $ i.e. $ \Phi_{\sigma^*} $ is a Schur matrix satisfying: $ \Phi^T_{\sigma^*} P \Phi_{\sigma^*} -e^{(-\beta |\sigma^* |T)}P\prec0$.
		\end{itemize}
		\STATE
		\STATE\textbf{{ONLINE Procedure:} At each instant $ \tau_k $:}
		\STATE$ {T_{avg}}_{max}=\frac{\sum_{j\in\{\mathtt{S}_{\sigma^*}^j = 0\}}1+|\sigma^* | }{|\mathcal{S}||\sigma^* |} $
		\STATE $ {\sigma_{opt}}=\{\sigma^*\} $
		\FORALL{$ \sigma \subset S^{l_{max}}_{l_{min}}(\mathcal{S}) $:}
		\IF{$ x^T_k\Big(\Phi^T_{\sigma} P \Phi_{\sigma} -e^{(-\beta |\sigma |T)}P\Big)x_k < 0 $} 
		\STATE $ T_{avg}=\frac{\sum_{j\in\{\mathtt{S}_{\sigma}^j = 0\}}1+|\sigma | }{|\mathcal{S}||\sigma |} $
		\STATE $ {\sigma_{opt}}=\{\sigma\}$ 
		\IF {$ T_{avg} >{T_{avg}}_{max}$} 
		\STATE$ {T_{avg}}_{max}= T_{avg}$
		\STATE $ {\sigma_{opt}}=\{\sigma\} $
		\ELSIF {$ T_{avg} ={T_{avg}}_{max}$}
		\STATE $ {\sigma_{opt}}=\{\sigma_{opt},\sigma\} $ 		 
		\ENDIF		 
		\ENDIF		  	 
		\ENDFOR
		\STATE Choose randomly  $ \sigma_k$, from $ {\sigma_{opt}} $.	
	\end{algorithmic}
\end{algorithm}

\subsection{Proposed Self Triggering Mechanism and Stability Analysis: Offline case}\label{wp-online}
A main drawback of the algorithm which is introduced in the previous is its heavy computation load. In other words, at each instant $ \tau_k $, Algorithm \ref{Alg1} should be performed to find the solution of the optimization problem. Computationally speaking, it is not applicable in a real-time implementation. For example, Algorithm \ref{Alg1} needs to examine the set of all horizons to find stable sequences and accordingly all optimal sequences. Thus, the algorithm has a complexity  of $ \mathcal{O}\Big\{\Big|S_{l_{min}}^{l_{max}}(\mathcal{S})\Big|=\sum_{i=l_{min}}^{l_{max}} \Big|\mathcal{S}\Big|^i\Big\} $, which is not realistic for a real-time implementation.

In this section, an offline tractable method for partitioning of the state space into a finite number of conic regions $ \mathcal{R}_c,~c \in \{1,\cdots,N\} $, is proposed. The idea is to define a set $ \Psi_{(c)} $ of optimal sequences for each of these regions. For this purpose, we use conic covering technique proposed in \cite{fiter2012state}:

Consider  $ N \in \mathbb{N}$ conic regions $ \mathcal{R}_c$ such that
\begin{equation}\label{conicreg}
	\mathcal{R}_c=\Big\{x \in \mathbb{R}^n: x^T Q_c x \geq 0 \Big\} ,~\forall c \in \{1,\cdots,N\},
\end{equation}
where the matrices $ Q_c $ are designed as it is explained in \cite{fiter2012state}.

A partition of the state space based on this technique is shown in Fig.~{\ref{Region1}}. 

\begin{figure}[!htpb] 
	\centering
	\subfloat{%
		\resizebox*{8cm}{!}{\includegraphics{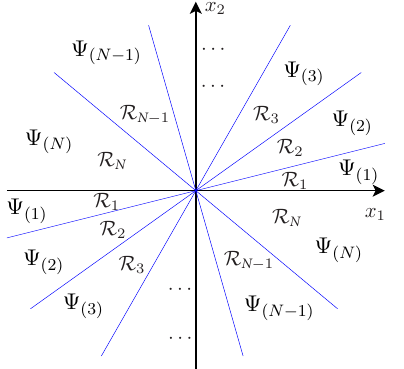}}}\hspace{2pt}
	\caption{ Covering the state space of dimension 2 with $ N $ conic regions $ \mathcal{R}_c, c \in \{1,\cdots,N\}  $} 	\label{Region1}
\end{figure}
\subsubsection{Self Triggering Mechanism and Stability Analysis}
Proposition \ref{STM_on} ensures exponential stability of the discretized system \eqref{Hdsyst} in which for a conic region $ \mathcal{R}_c $, the next sampling horizon $ \sigma^{(c)} $ defined by proposed self-triggering for each conic region $ \mathcal{R}_c,~c \in \{1,\cdots,N\} $. 

\begin{proposition}\label{STM_on}
	Consider a scalar $ \beta>0 $, a constant sampling interval $ T $, and a matrix $ P=P^T\succ0 $ such that	
	\begin{equation} \label{matrixp-off}
		\begin{split}
			& \exists \sigma^* \in S^{l_{max}}_{l_{min}}(\mathcal{S}): \Phi^T_{\sigma^*} P \Phi_{\sigma^*} -e^{(-\beta|\sigma^* |T) }P\prec0.
		\end{split} 
	\end{equation}
	
	Then the system \eqref{Hdsyst} with the self-triggering mechanism defined by:	
	\begin{equation} \label{STM-off-wp}
		\begin{split}
			& \sigma_{k}  \in \argmax_{\sigma \in \Psi_{(c)}:~ c \in \{1,\cdots,N\}~\text{and}~ x_k \in \mathcal{R}_c} \Bigg\{\frac{\sum_{j\in\{\mathtt{S}_{\sigma}^j = 0\}}1+|\sigma | }{|\mathcal{S}||\sigma |} \Bigg\},
		\end{split}
	\end{equation}
	with
	\begin{equation} 
		\begin{split}
			& \Psi_{(c)}  = \argmax_{{\sigma} \in {{\bar{S}}^{l_{max}}_{l_{min}}}(\mathcal{S},\mathcal{R}_c)}\Bigg\{ \frac{\sum_{j\in\{\mathtt{S}_{\sigma}^j = 0\}}1+|\sigma | }{|\mathcal{S}||\sigma |} \Bigg\},
		\end{split}
	\end{equation}
	where	
	\begin{equation}\label{es-off-wp}
\begin{split}
			{{\bar{S}}^{l_{max}}_{l_{min}}}(\mathcal{S},\mathcal{R}_c) =\bigg \{& \sigma \in S^{l_{max}}_{l_{min}}(\mathcal{S}) : \exists\epsilon_c>0 ,\\
			& \quad \Phi^T_{\sigma} P \Phi_{\sigma} -e^{(-\beta|\sigma |T) }P+\epsilon_c Q_c   \preceq 0 \bigg\},
\end{split}
	\end{equation}
	is exponentially stable with the decay rate of $ \beta/2 $.
\end{proposition}

\begin{proof}
	Consider a quadratic Lyapunov function $ V:\mathbb{R}^n  \to \mathbb{R}^+ $ defined by $ V(x_k)=x^T_k P x_k $ in which the matrix $ P $ satisfies \eqref{matrixp-off}. Then by construction, the set $ {{\bar{S}}^{l_{max}}_{l_{min}}}(\mathcal{S},\mathcal{R}_c) $ defined by \eqref{es-off-wp} is not empty (i.e. $ \sigma^* \in {\bar{S}}^{l_{max}}_{l_{min}}(\mathcal{S},x_k),~ \forall x_k \in \mathbb{R}^n $). Therefore,  the arguments of the maxima in the equation \eqref{STM-off-wp} is well-defined and non-empty. Let $ x_k \in \mathbb{R}^n, k\in \mathbb{N} $, There exists a conic region $ \mathcal{R}_c $ as in \eqref{conicreg} such that $ x_k \in  \mathcal{R}_c$. Consider the sampling horizon $ \sigma_k $ as defined in \eqref{STM-off-wp}. Therefore, we have $ x^T_{k} \Big(\Phi^T_{\sigma_k} P \Phi_{\sigma_k} -e^{(-\beta|\sigma_k |T) }P  \Big)x_{k} \leq x^T_{k} \Big(\Phi^T_{\sigma_k} P \Phi_{\sigma_k} -e^{(-\beta|\sigma_k |T) }P +\epsilon_c Q_c \Big)x_{k} \leq 0$, since by construction $ x_k\in \mathcal{R}_c $ and $ \sigma_k \in {{\bar{S}}^{l_{max}}_{l_{min}}}(\mathcal{S},\mathcal{R}_c) $. Then, one can conclude that
	
	\begin{equation}
		\begin{split}
			V(x_{k+1})&=x^T_{k+1}Px^T_{k+1}=x^T_k \Phi^T_{{\sigma}_k} P\Phi_{{\sigma}_k} x_k \\
			& \leq  e^{(-\beta|\sigma |T) }x^T_k P x_k= e^{-\beta (\tau_{k+1}-\tau_{k})} V(x_k).
		\end{split}
	\end{equation}
	Therefore, $ \forall k\in\mathbb{N} $, by recursion, we have $ V(x_k) \leq e^{-\beta \tau_k} V(x_0) $ and according to \cite{khalil2002nonlinear}, system \eqref{Hdsyst} with self triggering control \eqref{STM-off-wp} is exponentially stable with the decay rate of $ \beta/2 $.  
\end{proof}

\begin{remark}
	This should be noted that in this theorem, the larger the number of regions by which the state space is partitioned, the greater the performance the system will have.
\end{remark}

\subsubsection{Implementation Algorithm}
Algorithm \ref{Alg2} provides a self triggering algorithm for the sampled-date system \eqref{Hdsyst} which guarantees exponential stability and generates the next stabilizing optimal sampling horizon $ {\sigma} \in \Psi_{(c)} $ based on $ x(\tau_k) \in \mathcal{R}_c $ such that the selected sampling horizon has a maximum average of the sampling intervals:

\begin{algorithm}
	\caption{Self Triggering Mechanism for Sampled-Data System \eqref{Hdsyst}: Offline case}\label{Alg2}
	\begin{algorithmic}[H]
		\STATE
		\STATE \textbf{{OFFLINE Procedure:}}
		\begin{itemize}
			\STATE Define $ N $, the number of conic regions.
			\STATE Define $ l_{min} $ and $ l_{max}  $, the minimum and maximum length of the considered sampling horizons.
			\STATE Define a constant sampling interval $ T $ and  a finite set $ \mathcal{S}=\{\mathtt{S}^1, \mathtt{S}^2,\cdots,\mathtt{S}^m\} $ of sensors' number.
			\STATE Design the set $ S^{l_{max}}_{l_{min}}(\mathcal{S}) $ of sampling horizons of length $l \in 	[l_{min},l_{max}] $ constituted of sensors in $ \mathcal{S} $.
			\STATE Compute the matrix $ P=P^T\succ0 $ for a given stable sequence $ \sigma^*\in S^{l_{max}}_{l_{min}}(\mathcal{S}) $ i.e. $ \Phi_{\sigma^*} $ is a Schur matrix satisfying: $ \Phi^T_{\sigma^*} P \Phi_{\sigma^*} -e^{(-\beta|\sigma^* |T) }P\prec0$.
			\STATE Design a matrix $ Q_c $ for each conic region $ \mathcal{R}_c=\big\{x \in \mathbb{R}^n: x^T Q_c x \geq 0 \big\},~\forall c \in \{1,\cdots,N\}  $.
			\STATE  Compute the set $ \Psi_{(c)} $ of optimal sampling horizons for each conic region $ \mathcal{R}_c $:
		\end{itemize}
		
		{\renewcommand\labelitemi{}
			\begin{itemize}	
				\FORALL {$c=1:N$} 
				\STATE $ {T_{avg}}_{max}=\frac{\sum_{j\in\{\mathtt{S}_{\sigma^*}^j = 0\}}1+|\sigma^* | }{|\mathcal{S}||\sigma^* |} $
				\STATE $ {\Psi^{opt}_{(c)}}=\{\sigma^*\} $  
				\FORALL{$ \sigma \subset S^{l_{max}}_{l_{min}}(\mathcal{S}) $:}
				\IF{$ \Phi^T_{\sigma} P \Phi_{\sigma} -e^{(-\beta|\sigma |T) }P+\epsilon_c Q_c\prec0 $} 
				\STATE $ T_{avg}=\frac{\sum_{j\in\{\mathtt{S}_{\sigma}^j = 0\}}1+|\sigma | }{|\mathcal{S}||\sigma |}$ 
				\IF {$ T_{avg} > {T_{avg}}_{max}$} 
				\STATE $ {T_{avg}}_{max}= T_{avg}$
				\STATE $ {\Psi^{opt}_{(c)}}=\{\sigma\} $
				\ELSIF {$ T_{avg} ={T_{avg}}_{max}$}
				\STATE $ {\Psi^{opt}_{(c)}}=\{\Psi^{opt}_{(c)},\sigma\} $
				\ENDIF
				\ENDIF	 	 
				\ENDFOR
				\STATE $ \Psi_{(c)}={\Psi^{opt}_{(c)}} $ 
				\ENDFOR
		\end{itemize}}
		\STATE
		\STATE\textbf{{ONLINE Procedure:} At each instant $ \tau_k $:}
		\begin{itemize}
			\STATE Determine the conic region $ \mathcal{R}_c,~c \in \{1,\cdots,N\} $, in which current state $ x_k $ belongs.
			\STATE Select randomly a corresponding optimal horizon $ \sigma_k $ from the sets $ \Psi_{(c)},~c \in \{1,\cdots,N\} $, for which $ x_k \in \mathcal{R}_c $. 
		\end{itemize}
	\end{algorithmic}
\end{algorithm}

\section{Main Results: Perturbed Case}\label{Perturbed}
In this section, a self-triggered mechanism is introduced which maximizes the average of next sampling intervals regarding given set of user-defined inter-sampling periods for a continouse time LTI system under bounded perturbation. 
\subsection{System Reformulation}\label{sysref}
We consider a perturbed continuous time LTI system
\begin{equation} \label{systP}
	\begin{split}
		\dot{x}(t)&=Ax(t)+Bu(t)+Dw(t),\\
		t& \geq 0,~ x \in \mathbb{R}^n, ~ u \in \mathbb{R}^m, ~ w \in \mathbb{R}^{n_w}, 
	\end{split}
\end{equation}
where $ w(t) $ is a bounded perturbation. As before, we suppose that the state $ x(t) $ of the system \eqref{systP} is sampled at instants $ t_h $ with $ h \in \mathbb{N}  $ and the sampled-date sate-feedback controller satisfies \eqref{samplint} and \eqref{statefc}. Then, the discrete model of the system at sampling instants $ t_h,~h \in \mathbb{N} $ is given by
\begin{equation}\label{AdsystP}
	\begin{split}
		x(t_h+T)=A_{(T)}x(t_h)+B_{(T)} K \hat{x}(t_h)+\tilde{w}_{(T)}(t_h),
	\end{split}
\end{equation}
with 
\begin{equation}\label{dist_disc}
\begin{split}
		A_{(T)} &= e^{AT},\\
	B_{(T)} &=\int\nolimits_{0}^{T}e^{As}Bds,~ \tilde{w}_{(T)}(t_h)=  \int\nolimits_{0}^{T}e^{As}Dw(t_h+s) ds.
\end{split}
\end{equation}

\begin{assumption}\label{pertbounded}
	We suppose that the perturbation is bounded: 
	
	\begin{equation}
		\Big\| \tilde{w}_{(T)}(t_h)\Big\|_2 \leq \varpi,~ h\in \mathbb{N}.
	\end{equation} 
\end{assumption}

According to the notation used in the section \ref{sys-refor} and after first step, one can obtain, 

\begin{equation}\label{1stdist}
	\begin{split}
		x(t_{h+1}) &={A}_{(T)}x(t_h)+{B}_{(T)} K \hat{x}(t_h)+\tilde{w}_{(T)}(t_h)
	\end{split}
\end{equation}

In collective form, the \eqref{1stdist} can be written as 
\begin{equation}\label{1stSysC_dist}
	\eta_{h+1}=\tilde{A}_{(\mathtt{S}_h^1)}\eta_{h}+
	\begin{bmatrix}
		\tilde{w}_{(T)}(t_h)\\
		\textbf{0}
	\end{bmatrix}
	,~ \forall t_h=hT,~h\in \mathbb{N},
\end{equation}
with
\begin{equation*}
	\eta_{h+1}=[x(t_{h+1}) ~ \hat{x}(t_{h})]^T,~\eta_h=[x(t_h) ~ \hat{x}(t_{h-1})]^T,
\end{equation*} 
and
\begin{equation}\label{1stdiscollective_dist}
	\tilde{A}_{(\mathtt{S}_h^1)}= \begin{bmatrix}
		{A}_{(T)}+B_{(T)} K M(t_h) & B_{(T)} K  N(t_h)\\
		M(t_h) & N(t_h)
	\end{bmatrix}.
\end{equation}

After second step

\begin{equation}\label{2nddist}
	\begin{split}
		x(t_h+2T) &={A}_{(T)}x\Big(t_h+T\Big)+{B}_{(T)} K \hat{x}\Big(t_h+T\Big) \\
		& \quad +\tilde{w}_{(T)}\Big(t_h+T\Big)\\
		&={A}_{(T)}x\Big(t_h+T\Big)\\
		&\quad +{B}_{(T)} K M\Big(t_h+T\Big)x\Big(t_h+T\Big) \\
		& \quad +{B}_{(T)} K N\Big(t_h+T\Big)\hat{x}(t_h) +\tilde{w}_{(T)}\Big(t_h+T\Big)
	\end{split}
\end{equation}

In collective form, the \eqref{2nddist} can be written as 
\begin{equation}\label{2ndSysC_dist}
	\eta_{h+2}=\tilde{A}_{(\mathtt{S}_h^2)}\eta_{h+1}+
	\begin{bmatrix}
		\tilde{w}_{(T)}\Big(t_h+T\Big)\\
		\textbf{0}
	\end{bmatrix}
	,~ \forall t_h=hT,~h\in \mathbb{N},
\end{equation}
with
\begin{equation*}
	\eta_{h+2}=[x(t_{h+2}) ~ \hat{x}(t_{h+1})]^T,~\eta_h=[x(t_{h+1}) ~ \hat{x}(t_{h})]^T,
\end{equation*} 
and
\begin{equation}\label{2nddiscollective_dist}
	\tilde{A}_{(\mathtt{S}_h^2)}= \begin{bmatrix}
		{A}_{(T)}+B_{(T)} K M\Big(t_h+T\Big) & B_{(T)} K  N\Big(t_h+T\Big)\\
		M\Big(t_h+T\Big) & N\Big(t_h+T\Big)
	\end{bmatrix}.
\end{equation}

Combining equations \eqref{1stSysC_dist} and \eqref{2ndSysC_dist} yields
\begin{equation}\label{2nd_dist}
\begin{split}
		\eta_{h+2} &=\tilde{A}_{(\mathtt{S}_h^2)}\tilde{A}_{(\mathtt{S}_h^1)}\eta_{h} \\
		& \quad +\tilde{A}_{(\mathtt{S}_h^2)}\begin{bmatrix}
		\tilde{w}_{(T)}(t_h)\\
		\textbf{0}
	\end{bmatrix}+
	\begin{bmatrix}
		\tilde{w}_{(T)}\Big(t_h+T\Big)\\
		\textbf{0}
	\end{bmatrix},\\
	\forall t_h& =hT,~h\in \mathbb{N},
\end{split}
\end{equation}

After $ l_h \geq 1 $ steps, we have

\begin{equation}
	\begin{split}
		\eta(t_h+l_hT)	&=\Bigg(\prod_{s=1}^{l_h}{\tilde{A}}_{(\mathtt{S}_h^s)} \Bigg) \eta(t_h) \\
		& \quad +\Bigg(\sum_{i=0}^{l_h-1}\Bigg(\prod_{j=i+2}^{l_h}{\tilde{A}}_{(\mathtt{S}_h^j)}\Bigg)\begin{bmatrix}
			\tilde{w}_{(T)}\Big(t_h+iT\Big)\\
			\textbf{0}
		\end{bmatrix}\Bigg)
	\end{split}
\end{equation}
where 

\begin{equation*}
	\tilde{A}_{(\mathtt{S}_h^s)} = \begin{bmatrix}
		\begin{aligned}
			A_{(T)} + B_{(T)} K \times \\
			M\big(t_h + (s-1)T\big)
		\end{aligned}
		& \begin{aligned}
			B_{(T)} K \times \\
			N\big(t_h + (s-1)T\big)
		\end{aligned} \\[4ex]
		M\big(t_h + (s-1)T\big) & N\big(t_h + (s-1)T\big)
	\end{bmatrix}.
\end{equation*}

The representation of system \eqref{AdsystP} over a sampling horizon  $\sigma_k $ is given by
\begin{equation}\label{CLSP}
	\begin{split}
		{\eta}{_{k+1}} &=\Phi_{\sigma_k} \eta_k \\
		&\quad +\Bigg(\sum_{i=0}^{|\sigma_k|-1}\Bigg(\prod_{j=i+2}^{|\sigma_k|}{\tilde{A}}_{(\mathtt{S}_{\sigma_k}^j)}\Bigg)\begin{bmatrix}
			\tilde{w}_{(T)}\Big(\tau_k+iT\Big)\\
			\textbf{0}
		\end{bmatrix} \Bigg) \\
		&\quad =\Phi_{\sigma_k} \eta_k+\bar{w}_k,
	\end{split}
\end{equation}
in which $ x_k=x(\tau_k) $ and the transition matrix corresponding to the sampling horizon $ \sigma_k $, from instant $ \tau_k $ to instant $ \tau_{k+1} $ is given by 
\begin{equation}
	\Phi_{\sigma_k}={\tilde{A}}_{(\mathtt{S}^{l_k}_{\sigma_k})}{\tilde{A}}_{(\mathtt{S}^{l_{k-1}}_{\sigma_k})}\cdots {\tilde{A}}_{(\mathtt{S}^{l_{1}}_{\sigma_k})}
\end{equation} 
where 
\begin{equation}
	\tau_{k+1}=\tau_{k}+l_{k}T.
\end{equation} 
Since  $ {\tilde{A}}_{(\mathtt{S}_{\sigma_k}^s)},~\forall s,k \in \mathbb{N}  $, is a bounded operator and the set $ \mathcal{S}  $ of considered sampling intervals is finite, there exists a constant $ C $ such that 
\begin{equation}
	C=\max_{\mathtt{S}_{\sigma_k}^s \in \mathcal{S}} \Big(\Big\|{\tilde{A}}_{(\mathtt{S}_{\sigma_k}^s)}\Big\|_2\Big),~\forall k,s \in \mathbb{N},
\end{equation}
According to Assumption \ref{pertbounded},  one can conclude that 

\begin{equation}
\begin{aligned}
		\Big\|\bar{w}_k\Big\|_2 & = \Bigg\|\Bigg(\sum_{i=0}^{|\sigma_k|-1}\Bigg(\prod_{j=i+2}^{|\sigma_k|}{\tilde{A}}_{(\mathtt{S}_{\sigma_k}^j)}\Bigg)\begin{bmatrix}
		\tilde{w}_{(T)}\Big(\tau_k+iT\Big)\\
		\textbf{0}
	\end{bmatrix} \Bigg)\Bigg\|_2 \\
	& \leq \varpi \sum_{q=0}^{|\sigma_k|-1}C^q .
\end{aligned}
\end{equation}
\begin{definition}\cite{khalil2002nonlinear}
	Consider system 
	\begin{equation}\label{exbound}
		\dot{x}=f(x) 
	\end{equation} 
	where $ f:\mathbb{R}^n \to \mathbb{R}^n $ is locally Lipschitz in $ x $. 
	
	The solution of  \eqref{exbound} is Globally Uniformly Ultimately Bounded (GUUB) with ultimate bound $ b>0 $ independent of $ t_0 \geq 0 $, and for every arbitrary large $ a > 0 $, there is $ T=T(a,b)\geq 0 $, independent of $ t_0 $, such that
	\begin{equation}\label{exbdef}
		\Big\|x(t_0) \Big\| \leq a  \Rightarrow \Big\|x(t)\Big\| \leq b, \quad \forall t \geq t_0+T
	\end{equation}
\end{definition}\label{defUBB}

\subsection{Problem Statement}
Consider closed-loop system \eqref{CLSP}, our aim is to calculate for $ x_k $ at instant $ \tau_k $, the next optimal sampling horizon $ \sigma_k $ which will be applied to the sampling mechanism in order to ensure global uniform ultimate boundedness of the system's solution.

\subsection{Proposed self triggering mechanism and Stability Analysis: Online case}\label{STM-off-p}

\subsubsection{Self Triggering Mechanism and Stability Analysis}
Proposition \ref{STMS-off-p} ensures global uniform ultimate boundedness of the discretized system \eqref{CLSP} in which the next sampling horizon $ \sigma_k $ defined by proposed self-triggering mechanism.

In the rest of the paper, we use the notation $ \mathcal{E}(P,\upsilon)$ to refer to the ellipsoid
\begin{equation}
	\mathcal{E}(P,\upsilon)=\bigg\{x \in \mathbb{R}^n: x^T Px \leq  \upsilon	\bigg\}.
\end{equation}

For a positive definite matrix $ P \in \mathbb{R}^{n \times n} $ and a positive scalar $ \upsilon $. For all positive scalar $ r $, we denote by $ \mathcal{B}(0,r) $ the ball of radius $ \sqrt{r} $:
\begin{equation}\label{ball}
	\mathcal{B}(0,r)=\bigg\{x \in \mathbb{R}^n: 
	\Big\| x \Big\|^2_2 \leq r	\bigg\}.
\end{equation}
\begin{proposition}\label{STMS-off-p}
	Consider a constant sampling period $ T $, scalars $ \beta>0$ and $\gamma>0 $. If there exists $ \sigma^* \in S^{l_{max}}_{l_{min}}(\mathcal{S}) $ and  symmetric positive definite matrices $ P,~M$ such that 
	
	\begin{equation} \label{LMI1-STM-on-p}
		\begin{split}	
			&\Phi_{\sigma^*}^T(P+M)\Phi_{\sigma^*} +\Big(e^{(-\beta|{\sigma^*} |T) }-\gamma\Big)P \preceq 0,
		\end{split} 
	\end{equation}
	and
	\begin{equation} \label{LMI2-STM-on-p}
		\begin{bmatrix} M &  P\\ P &  \frac{\gamma}{\chi}I- P \end{bmatrix} \succeq 0, 
	\end{equation}
	are verified with $ \chi =\Big(\varpi  \sum_{q=0}^{|\sigma^*|-1}C^q \Big)^2 $, then the system \eqref{CLSP} is GUUB with the self-triggering mechanism defined by
	

\begin{equation}\label{STM-on-p}
	\sigma_k \in
	\begin{cases}
		\begin{multlined}
			\argmax_{\sigma \in {\bar{S}}^{l_{max}}_{l_{min}}(\mathcal{S},x_k)} \\
			\left\{ \frac{\sum_{j\in\{\mathtt{S}_{\sigma}^j = 0\}}1+|\sigma | }{|\mathcal{S}||\sigma |} \right\},
		\end{multlined}
		& \text{If } x_k \notin \mathcal{E}(P,1), \\
		\Big\{\Big(0\Big)\Big\}, & \text{If } x_k \in \mathcal{E}(P,1).
	\end{cases}
\end{equation}

	with
	
	\begin{equation}\label{es-on-p}
		\begin{split}
			&{{\bar{S}}^{l_{max}}_{l_{min}}}(\mathcal{S},x_k) =\Bigg \{\sigma \in S^{l_{max}}_{l_{min}}(\mathcal{S});~\begin{pmatrix}
				x_k \\
				1
			\end{pmatrix}^T U_{\sigma} \begin{pmatrix}
				x_k \\
				1
			\end{pmatrix} \geq 0 \Bigg \},
		\end{split}
	\end{equation}
	\begin{equation} \label{U_sigma}
		\begin{split}
			U_{\sigma}=\text{diag}\Bigg(-\Phi^T_{{\sigma}} \Big(P+M\Big)\Phi_{{\sigma}}+\Big(e^{(-\beta|{\sigma} |T) }-\gamma\Big) P \\
			,-\Big(\varpi  \sum_{q=0}^{|\sigma|-1}C^q \Big)^2\lambda_{max}(P M^{-1}P+P)+\gamma  \Bigg),
		\end{split}
	\end{equation}
	Furthermore, the solution of \eqref{CLSP} converges and remains in $ \mathcal{E}(P,\mu) $, with 
	\begin{equation}
		\mu= \lambda_{max}(P) \Big(\frac{C'}{\lambda_{min}(P)}+ \varpi\Big)^2 ,
	\end{equation} 
	where $ C'=\Big\| \tilde{A}_{(0)} \Big\|_2 $  and $ \varpi $ is given in Assumption \ref{pertbounded}.	  
\end{proposition}

\begin{proof}
	See Appendix.
\end{proof}

\begin{remark}
	To obtain the smallest ellipsoid $ \mathcal{E}(P,\mu) $ by which the solution of \eqref{CLSP} is bounded, one can consider a ball $ \mathcal{B}(0,\psi) $ tangent to the ellipsoid $ \mathcal{E}(P,\mu) $ (See Fig.~{\ref{MatP2}}) such that if there exists $ \sigma^* \in S^{l_{max}}_{l_{min}}(\mathcal{S}) $ and  symmetric positive definite matrices $ P,~M$ such that 
	\begin{equation}\label{LMI1_Per_On} 
		\begin{split}	
			&\Phi_{\sigma^*}^T(P+M)\Phi_{\sigma^*} +\Big(e^{(-\beta|{\sigma^*} |T) }-\gamma\Big)P \succeq 0,
		\end{split} 
	\end{equation}
	and
	\begin{equation}\label{LMI2_Per_On} 
		\begin{bmatrix} M &  P\\ P &  \frac{\gamma}{\chi}I- P \end{bmatrix} \succeq 0, 
	\end{equation}
	are verified and for the horizon $ \sigma^* $:
	
	\begin{equation}
		\begin{split}
			&\min_{\psi \in \mathbb{R}\setminus \{0\}} \psi \\
			& \text{s.t.}~ P \succ \underline{\zeta}I,~\underline{\zeta} > \frac{\mu}{\psi}.
		\end{split}
	\end{equation}
	\begin{figure}[!h] 
		\centering
		\subfloat{%
			\resizebox*{5cm}{!}{\includegraphics{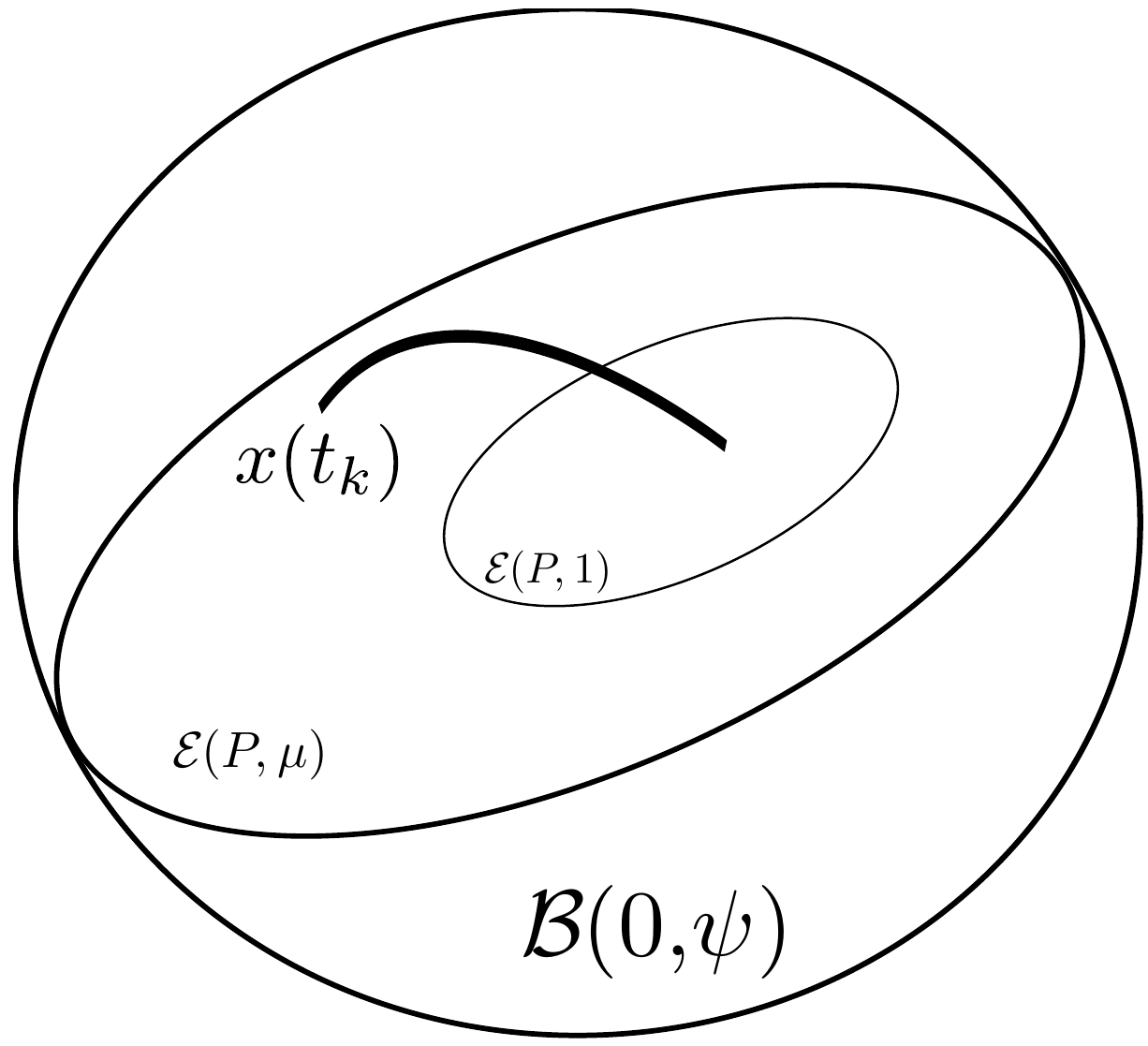}}}\hspace{2pt}
		\caption{ The illustration of three regions $ \mathcal{E}(P,1) $, $ \mathcal{E}(P,\mu) $, and $ \mathcal{B}(0,\psi) $.} 	\label{MatP2}
	\end{figure}
\end{remark}

\subsubsection{Implementation Algorithm}
Algorithm \ref{Alg3} provides a self triggering algorithm for the sampled-date system \eqref{CLSP} which generates a next optimal sampling horizon $ {\sigma}_k $ based on $ x(\tau_k) $ such that the selected sampling horizon ensures the system's solution is GUUB.

\begin{algorithm}
	\caption{Self Triggering Mechanism for Sampled-Data System \eqref{CLSP}: Online case}\label{Alg3}
	\begin{algorithmic}
		\STATE
		\STATE \textbf{{OFFLINE Procedure:}}
		\begin{itemize}
			\STATE Define $ l_{min} $ and $ l_{max}  $, the minimum and maximum length of the considered sampling horizons.
			\STATE Select positive constant $ \gamma $.
			\STATE Define a constant sampling interval $ T $ and  a finite set $ \mathcal{S}=\{\mathtt{S}^1, \mathtt{S}^2,\cdots,\mathtt{S}^m\} $ of sensors' number.
			\STATE Design the set $ S^{l_{max}}_{l_{min}}(\mathcal{S}) $ of sampling horizons of length $l \in 	[l_{min},l_{max}] $ constituted of sensors in $ \mathcal{S} $.
			\STATE Compute the matrix $ P=P^T \succ 0 $ and $ M=M^T\succ0 $ for a given stable sequence $ \sigma^*\in S^{l_{max}}_{l_{min}}(\mathcal{S}) $ i.e. $ \Phi_{\sigma^*} $ is a Schur matrix satisfying: 
			$\Phi_{\sigma^*}^T(P+M)\Phi_{\sigma^*} +\Big(e^{(-\beta|{\sigma^*} |T) }-\gamma\Big)P \succeq 0, \begin{bmatrix} M &  P\\  P & \frac{\gamma}{\chi} I- P \end{bmatrix} \succeq 0$ in which $ \chi =\Big(\varpi  \sum_{q=0}^{|\sigma^*|-1}C^q \Big)^2 $.
		\end{itemize}
		\STATE
		\STATE\textbf{{ONLINE Procedure:} At each instant $ \tau_k $:}
		\STATE$ {T_{avg}}_{max}=\frac{\sum_{j\in\{\mathtt{S}_{\sigma^*}^j = 0\}}1+|\sigma^* | }{|\mathcal{S}||\sigma^* |} $
		\STATE $ \sigma_{opt}=\{\sigma^*\} $
		\FORALL{$ \sigma \subset S^{l_{max}}_{l_{min}}(\mathcal{S}) $:}
		\IF{$ \begin{pmatrix}
				x_k \\
				1
			\end{pmatrix}^T U_{\sigma} \begin{pmatrix}
				x_k \\
				1
			\end{pmatrix}\geq 0 $} 
		\STATE $ T_{avg}=\frac{\sum_{j\in\{\mathtt{S}_{\sigma}^j = 0\}}1+|\sigma | }{|\mathcal{S}||\sigma |}$ 
		\IF {$ T_{avg} > {T_{avg}}_{max}$} 
		\STATE $ {T_{avg}}_{max}= T_{avg}$
		\STATE $ {\sigma_{opt}}=\{\sigma\} $
		\ELSIF {$ T_{avg} ={T_{avg}}_{max}$}
		\STATE $ {\sigma_{opt}}=\{\sigma_{opt},\sigma\} $ 
		\ENDIF
		\ENDIF	 	 
		\ENDFOR
		\STATE Choose randomly  $ \sigma_k $, in $ {\sigma_{opt}} $.	
	\end{algorithmic}
\end{algorithm}
\subsection{Proposed self triggering mechanism and Stability Analysis: Offline case}
Similar to the Section \ref{wp-online}, an offline tractable method is used for splitting of the state space into regions of user-defined number to compute an optimal possible stablizing horizon based on a region that the state of the system \eqref{CLSP} is belong to such that ensures globally uniformly ultimately  boundedness of the solution.
\subsubsection{Self Triggering Mechanism and Stability Analysis} 
Proposition \ref{STMS-on-p} ensures globally uniformly ultimately  boundedness of the discretized system \eqref{CLSP} in which the next optimal sampling horizon $ \Psi_{(c)} $ defined by proposed self-triggering for a conic region $ \mathcal{R}_c, c \in \mathbb{N} $.
\begin{proposition}\label{STMS-on-p}
	Consider scalars $ \beta>0$, $\gamma_1>0 $, $\gamma_2>0 $, and a constant sampling interval $ T $. If there exists $ \sigma^* \in S^{l_{max}}_{l_{min}}(\mathcal{S}) $ and  symmetric positive definite matrices $ P$ such that 
	\begin{equation} \label{SigmaOmega-off}
		\begin{split}
			U=&\begin{pmatrix}
				-\Phi_{\sigma^*}^TP\Phi_{\sigma^*}+\Big(\bar{\beta}-\gamma_1\Big)P & * & *\\
				-P\Phi_{\sigma^*} & \frac{\gamma_2}{\chi}I-P & *\\
				\textbf{0} & \textbf{0} &  {-\gamma_2+\gamma_1} 
			\end{pmatrix}\succeq 0 ,			
		\end{split} 
	\end{equation}
	in which $ \chi =\varpi \sum_{q=0}^{|\sigma^*|-1}C^q $ and $ \bar{\beta}=e^{(-\beta|{\sigma^*} |T) } $, then the system \eqref{CLSP} is GUUB with the self-triggering mechanism defined by

\begin{equation}\label{STMoff-p}
	\sigma_k \in
	\begin{cases}
		\begin{multlined}
			\argmax_{\substack{\sigma \in \Psi_{(c)}, \\ c \in \{1,\cdots,N\},\\ x_k \in \mathcal{R}_c}}
			\left\{ \frac{\sum_{j\in\{\mathtt{S}_{\sigma}^j = 0\}}1+|\sigma | }{|\mathcal{S}||\sigma |} \right\},
		\end{multlined}
		& \text{If } x_k \notin \mathcal{E}(P,1), \\
		\Big\{\Big(0\Big)\Big\}, & \text{If } x_k \in \mathcal{E}(P,1).
	\end{cases}
\end{equation}

	with
	\begin{equation} 
		\begin{split}
			& \Psi_{(c)}  = \argmax_{{\sigma} \in {{\bar{S}}^{l_{max}}_{l_{min}}}(\mathcal{S},\mathcal{R}_c)}\Bigg\{ \frac{\sum_{j\in\{\mathtt{S}_{\sigma}^j = 0\}}1+|\sigma | }{|\mathcal{S}||\sigma |}  \Bigg\},
		\end{split}
	\end{equation}
	where
	\begin{equation}\label{es-off-p}
		\begin{split}
			&{{\bar{S}}^{l_{max}}_{l_{min}}}(\mathcal{S},\mathcal{R}_c) =\Big\{\sigma \in S^{l_{max}}_{l_{min}}(\mathcal{S}),~\epsilon_c>0;~U_{c}\succeq 0 \Big\},
		\end{split}
	\end{equation}
	and $ U_{c}=[u_{ij}] $ is a symmetric matrix defined by
	\begin{equation}
		\begin{split}		
			u_{11}&=\epsilon_c Q_c-\Phi_{\sigma}^T P \Phi_{\sigma}+\bar{\beta} P-\gamma_1 P,\\
			u_{22} & =\frac{\gamma_2}{\varpi \sum_{q=0}^{|\sigma|-1}C^q}I-P,\\
			u_{33}& ={-\gamma_2+\gamma_1},\\
			u_{21}& =-P\Phi_{\sigma},\\
			u_{31}& =\textbf{0},\\
			u_{32}&=\textbf{0} .
		\end{split}
	\end{equation}
	Furthermore, the solution of \eqref{CLSP} converges and remains in $ \mathcal{E}(P,\mu) $, with 
	\begin{equation}
		\mu= \lambda_{max}(P) \Big(\frac{C'}{\lambda_{min}(P)}+ \varpi\Big)^2 ,
	\end{equation} 
	where $ C'=\Big\| \tilde{A}_{(0)} \Big\|_2 $  and $ \varpi $ is given in Assumption \ref{pertbounded}.	
	
\end{proposition}
\begin{proof}
	See Appendix.
\end{proof}
\subsubsection{Implementation Algorithm}
Algorithm \ref{Alg4} provides a self triggering mechanism for the sampled-date system \eqref{CLSP} which generates the next stabilizing optimal sampling horizon $ {\sigma} \in \Psi_{(c)} $ based on $ x(\tau_k) \in \mathcal{R}_c $ such that the selected sampling horizon ensures globally uniformly ultimately boundedness of the solution.
\begin{algorithm}
	\caption{Self Triggering Mechanism for Sampled Data System \eqref{CLSP}: Offline case}\label{Alg4}
	\begin{algorithmic}[H]
		\STATE
		\STATE \textbf{{OFFLINE Procedure:}}
		\begin{itemize}
			\STATE Define $ N $: number of conic regions.
			\STATE Define $ l_{max} $ and $ l_{min}  $, the minimum and maximum length of the considered sampling horizons.
			\STATE Select positive constants $ \gamma_1 $ and $ \gamma_2 $.
			\STATE Define a constant sampling interval $ T $ and  a finite set $ \mathcal{S}=\{\mathtt{S}^1, \mathtt{S}^2,\cdots,\mathtt{S}^m\} $ of sensors' number.
			\STATE Design the set $ S^{l_{max}}_{l_{min}}(\mathcal{S}) $ of sampling hoizons of length $l \in 	[l_{min},l_{max}] $ constituted of sensors in $ \mathcal{S} $.
			\STATE Compute the matrix $ P=P^T\succ0 $ for a given stable sequence $ \sigma^*\in S^{l_{max}}_{l_{min}}(\mathcal{S}) $ i.e. $ \Phi_{\sigma^*} $ is a Schur matrix satisfying: $ \begin{pmatrix}
				-\Phi_{\sigma^*}^TP\Phi_{\sigma^*}+\Big(\bar{\beta}-\gamma_1 \Big)P & * & *\\
				-P\Phi_{\sigma^*} & \frac{\gamma_2}{\chi}I_2-P & *\\
				\textbf{0} & \textbf{0} &  {-\gamma_2+\gamma_1} 
			\end{pmatrix}\succeq 0	$ in which $ \chi =\varpi \sum_{q=0}^{|\sigma^*|-1}C^q $ and $ \bar{\beta}=e^{(-\beta|{\sigma^*} |T) } $.
			\STATE Design a matrix $ Q_c $ for each conic region $ \mathcal{R}_c=\big\{x \in \mathbb{R}^n: x^T Q_c x \geq 0 \big\},~\forall c \in \{1,\cdots,N\}  $.
			\STATE  Compute the set $ \Psi_{(c)} $ of optimal sampling horizons for each conic region $ \mathcal{R}_c $:
		\end{itemize}
		{\renewcommand\labelitemi{}
			\begin{itemize}	
				\FORALL {$c=1:N$} 
				\STATE $ {T_{avg}}_{max}=\frac{\sum_{j\in\{\mathtt{S}_{\sigma^*}^j = 0\}}1+|\sigma^*| }{|\mathcal{S}||\sigma^*|} $
				\STATE $ {\Psi^{opt}_{(c)}}=\{\sigma^*\} $  
				\FORALL{$ \sigma \subset S^{l_{max}}_{l_{min}}(\mathcal{S}) $:}
				\IF{$ U \succeq 0 $} 
				\STATE $ T_{avg}=\frac{\sum_{j\in\{\mathtt{S}_{\sigma}^j = 0\}}1+|\sigma | }{|\mathcal{S}||\sigma |}$ 
				\IF {$ T_{avg} > {T_{avg}}_{max}$} 
				\STATE $ {T_{avg}}_{max}= T_{avg}$
				\STATE $ {\Psi^{opt}_{(c)}}=\{\sigma\} $
				\ELSIF {$ T_{avg} ={T_{avg}}_{max}$}
				\STATE $ {\Psi^{opt}_{(c)}}=\{\Psi^{opt}_{(c)},\sigma\} $
				\ENDIF
				\ENDIF	 	 
				\ENDFOR
				\STATE $ \Psi_{(c)}={\Psi^{opt}_{(c)}} $ 
				\ENDFOR
		\end{itemize}}	
		\STATE
		\STATE\textbf{{ONLINE Procedure:} At each instant $ \tau_k $:}
		\begin{itemize}
			\STATE Determine the  conic region $ \mathcal{R}_c,~c \in \{1,\cdots,N\} $, in which current state $ x_k $ belongs.
			\STATE Select randomly a corresponding optimal horizon $ \sigma_k $ from the sets $ \Psi_{(c)},~c \in \{1,\cdots,N\} $, for which $ x_k \in \mathcal{R}_c $. 
		\end{itemize}
	\end{algorithmic}
\end{algorithm}


\section{Simulation}\label{Simulation}
\subsection{Perturbation-free Case}
\subsubsection{Proposed Self Triggering Mechanism: Online Procedure}
To verify the performance of the self-triggering mechanism \eqref{STM-on-wp}, we consider a continuous time LTI control system \eqref{exp1}.  
Consider sampling horizon of minimum length $ l_{min}=1 $ and maximum length $ l_{max}=3 $, a set of sensors $ \Gamma=\{\begin{matrix}
	\mathtt{S}^1, \mathtt{S}^2
\end{matrix}\} $, and the initial state $ x_0=[\begin{matrix} 5& -2 & 5& -2\end{matrix}]$. In addition, we define a constant sampling interval $ T = 0.3 $. 

Evolution of the system’s states and Lyapunov function using variable sensors sampling with a decay rate $ \beta=0 $ are shown in Fig.{ \ref{state_wp_on}} and Fig.~{\ref{lyp_wp_on}}, respectively. 

It is observed that by employing the proposed mechanism, the system reduced sensor utilization by $ 73.58\% $ compared to the scenario in which the system uses the sensors in each time step. In addition, Sensors' status for this simulation is shown in Fig.~{\ref{sstatus_wp_on}}.

\begin{figure}[!htpb] 
	\centering
	\subfloat{%
		\resizebox*{8cm}{!}{\includegraphics{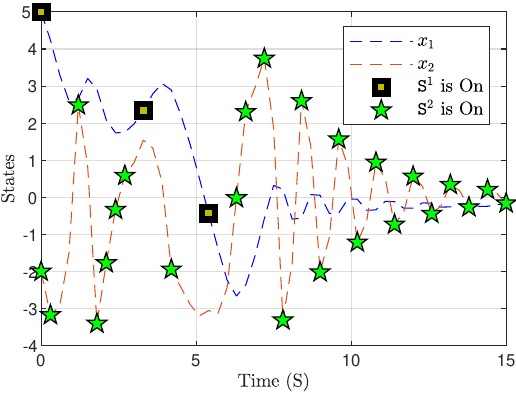}}}
	\caption{System's states.} 	\label{state_wp_on}
\end{figure} 

\begin{figure}[!htpb] 
	\centering
	\subfloat{%
		\resizebox*{8cm}{!}{\includegraphics{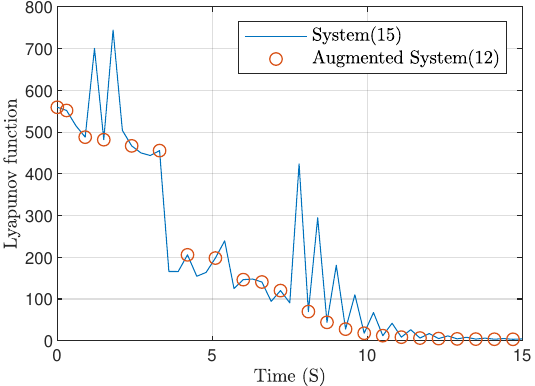}}}
	\caption{Lyapunov function.} 	\label{lyp_wp_on}
\end{figure} 

\begin{figure}[!htpb] 
	\centering
	\subfloat{%
		\resizebox*{9cm}{!}{\includegraphics{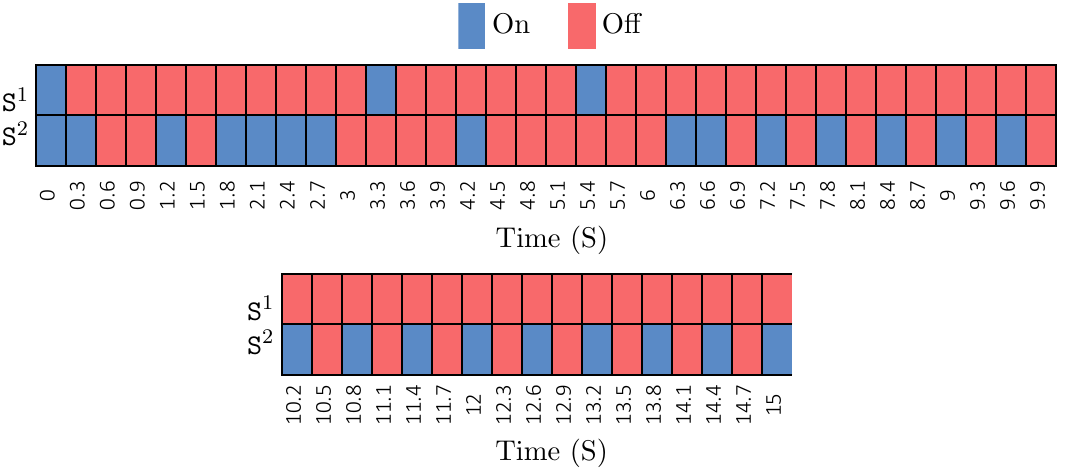}}}
	\caption{Sensors' status in each time step.} 	\label{sstatus_wp_on}
\end{figure}

\subsubsection{Proposed Self Triggering Mechanism: Offline Procedure}
We verify the performance of the self-triggering mechanism \eqref{STM-off-wp}, we consider a continuous time LTI control system \eqref{exp1}. Consider sampling horizon of minimum length $ l_{min}=1 $ and maximum length $ l_{max}=6 $, a set of sensors $ \Gamma=\{\begin{matrix}
	\mathtt{S}^1, \mathtt{S}^2
\end{matrix}\} $, and the initial state $ x_0=[\begin{matrix} 15& -1.5 & 15 & -1.5 \end{matrix}]$. In addition, we define a constant sampling interval $ T = 0.205 $. In addition, the number of region is $ N =15 $;

Evolution of the system’s states and Lyapunov function using variable sensors sampling with a decay rate $ \beta=0 $ are shown in Fig.{ \ref{state_wp_off}} and Fig.{ \ref{lyp_wp_off}}, respectively. 

It is observed that by employing the proposed mechanism, the system reduced sensor utilization by $ 70.27\% $ compared to the scenario in which the system uses the sensors in each time step. In addition, Sensors' status for this simulation is shown in Fig.~{\ref{sstatus_wp_off}}.

\begin{figure}[!htpb] 
	\centering
	\subfloat{%
		\resizebox*{8cm}{!}{\includegraphics{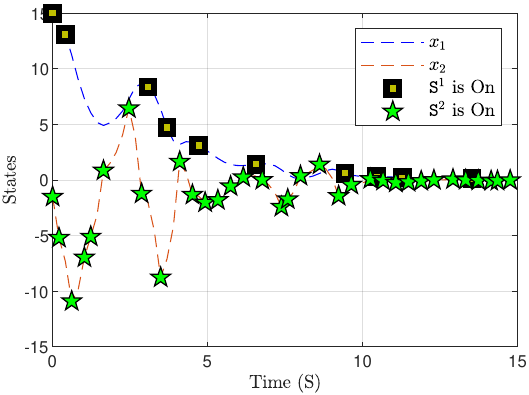}}}
	\caption{System's states.} 	\label{state_wp_off}
\end{figure} 

\begin{figure}[!htpb] 
	\centering
	\subfloat{%
		\resizebox*{8cm}{!}{\includegraphics{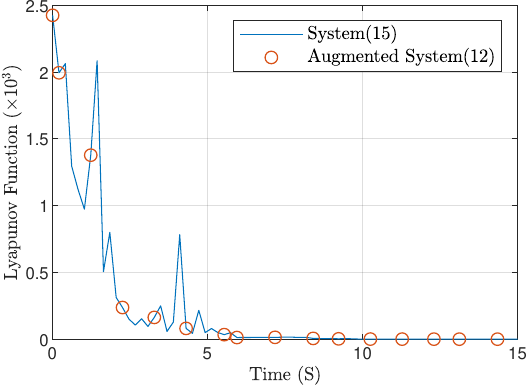}}}
	\caption{Lyapunov function.} 	\label{lyp_wp_off}
\end{figure} 

\begin{figure}[!htpb] 
	\centering
	\subfloat{%
		\resizebox*{9cm}{!}{\includegraphics{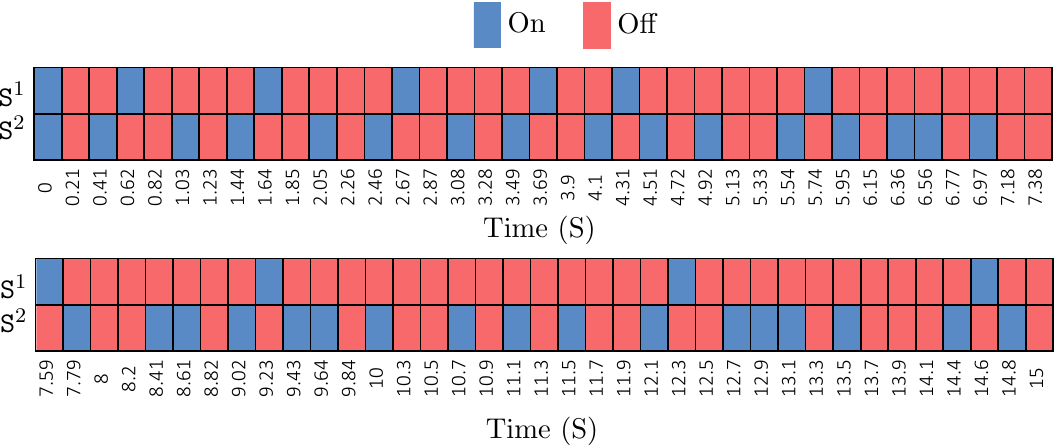}}}
	\caption{Sensors' status in each time step.} 	\label{sstatus_wp_off}
\end{figure}

\subsection{Perturbed Case}
\subsubsection{Proposed Self Triggering Mechanism: Online Procedure}
Consider a continuous time LTI control system described in \eqref{systP} with
\begin{equation}\label{sys_p}
\begin{split}
		A & =\begin{bmatrix}
		0&1\\-2&3
	\end{bmatrix},\\ 
	B & =\begin{bmatrix}
		0\\1
	\end{bmatrix}, \\
	K & =\begin{bmatrix}
		1&-4
	\end{bmatrix}, \\ 
	D & =\begin{bmatrix}
		1\\1
	\end{bmatrix},\\
	w(t)& =sin(5\pi t)
\end{split}
\end{equation}

Also, we choose sampling horizon of minimum length $ l_{min}=1 $ and maximum length $ l_{max}=6 $, a set of sensors $ \Gamma=\{\begin{matrix}
	\mathtt{S}^1, \mathtt{S}^2
\end{matrix}\} $, and the initial state $ x_0=[\begin{matrix} 5& -2 \end{matrix}]$.
The LMI in equations \eqref{LMI1_Per_On} and \eqref{LMI2_Per_On} with $ \gamma=0.35 $ yields, 
\begin{equation*}
	\begin{split}
		P&=\begin{bmatrix} 4.5107  & -0.3699 &  -0.7505  & -1.3990 \\ -0.3699  &  5.0824 &   0.4709  & -1.3114 \\ -0.7505  &  0.4709  &  6.2863  &  0.2754 \\
			-1.3990  & -1.3114 &   0.2754 &   2.0002 \end{bmatrix},\\  
		M&=\begin{bmatrix} 9.6164  & -0.2298  & -2.4446  & -2.3527 \\
			-0.2298  & 10.7068  &  1.3307 &  -3.6711\\
			-2.4446  &  1.3307  & 10.9375  &  0.6381\\
			-2.3527  & -3.6711 &   0.6381  &  4.5667 \end{bmatrix}
	\end{split}
\end{equation*}

Evolution of the system’s states and Lyapunov function using variable sampling intervals with a decay rate $ \beta=0 $ are shown in Fig.{ \ref{state_p_on}} and Fig.{ \ref{lyp_p_on}}, respectively. It is seen that the system used the both sensors $ 68.94\% $ less than the scenario in which the system uses the sensors in each time step. In addition, Sensors' status for this simulation is shown in Fig.{ \ref{sstatus_p_on}}.
\begin{figure}[!htpb] 
	\centering
	\subfloat{%
		\resizebox*{8cm}{!}{\includegraphics{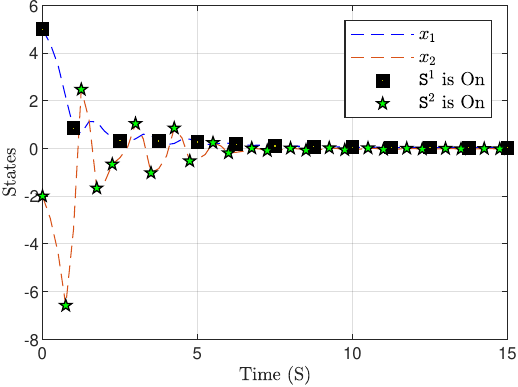}}}
	\caption{System's states.} 	\label{state_p_on}
\end{figure} 

\begin{figure}[!htpb] 
	\centering
	\subfloat{%
		\resizebox*{8cm}{!}{\includegraphics{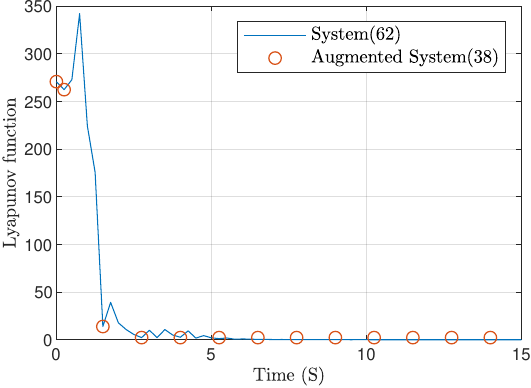}}}
	\caption{Lyapunov function.} 	\label{lyp_p_on}
\end{figure} 

\begin{figure}[!htpb] 
	\centering
	\subfloat{%
		\resizebox*{9cm}{!}{\includegraphics{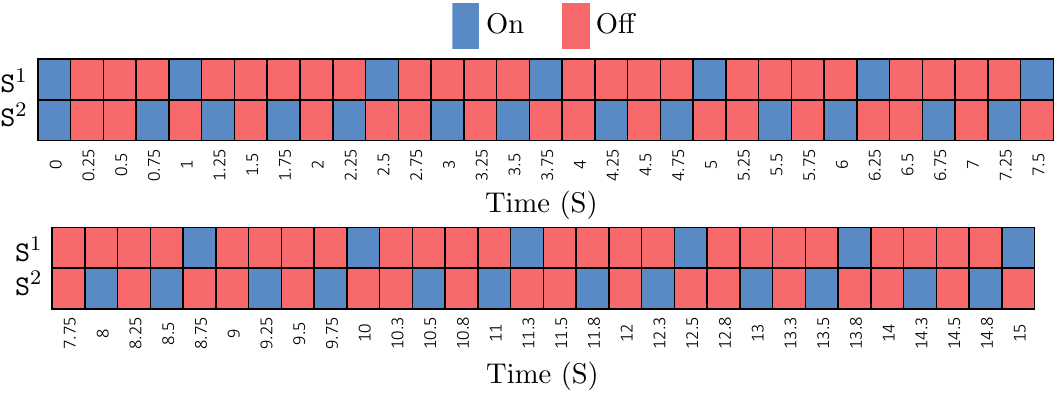}}}
	\caption{Sensors' status in each time step.} 	\label{sstatus_p_on}
\end{figure}

\subsubsection{Proposed Self Triggering Mechanism: Offline Procedure}
We verify the performance of the self-triggering mechanism \eqref{STMoff-p}, we consider a continuous time LTI control system \eqref{sys_p}. Consider sampling horizon of minimum length $ l_{min}=3 $ and maximum length $ l_{max}=6 $, a set of sensors $ \Gamma=\{\begin{matrix}
	\mathtt{S}^1, \mathtt{S}^2
\end{matrix}\} $, and the initial state $ x_0=[\begin{matrix} 15& -1.5 & 15 & -1.5 \end{matrix}]$. In addition, we define a constant sampling interval $ T = 0.205 $. In addition, the number of region is $ N =15 $;

Evolution of the system’s states and Lyapunov function using variable sensors sampling with a decay rate $ \beta=0 $ are shown in Fig.{ \ref{state_p_off}} and Fig.{ \ref{lyp_p_off}}, respectively. 

It is observed that by employing the proposed mechanism, the system reduced sensor utilization by $ 59.21\% $ compared to the scenario in which the system uses the sensors in each time step. In addition, Sensors' status for this simulation is shown in Fig.~{\ref{sstatus_p_off}}.

\begin{figure}[!htpb] 
	\centering
	\subfloat{%
		\resizebox*{8cm}{!}{\includegraphics{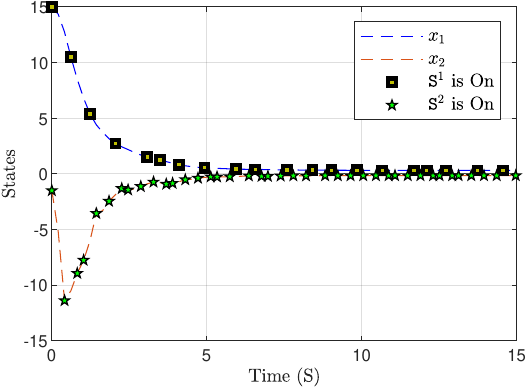}}}
	\caption{System's states.} 	\label{state_p_off}
\end{figure} 

\begin{figure}[!htpb] 
	\centering
	\subfloat{%
		\resizebox*{8cm}{!}{\includegraphics{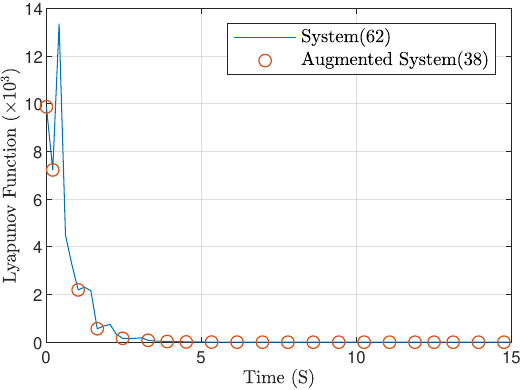}}}
	\caption{Lyapunov function.} 	\label{lyp_p_off}
\end{figure} 

\begin{figure}[!htpb] 
	\centering
	\subfloat{%
		\resizebox*{9cm}{!}{\includegraphics{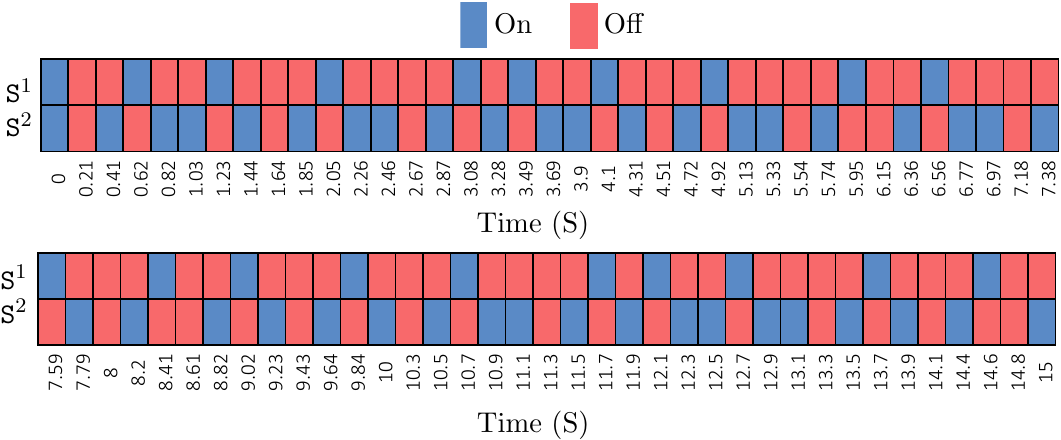}}}
	\caption{Sensors' status in each time step.} 	\label{sstatus_p_off}
\end{figure}

\section{Conclusion}\label{Conclusion}

This work develops self-triggered control for linear systems with asynchronous measurements. The controller computes optimal horizons (sensor selection sequences) at each sampling instant, maximizing inter-sample intervals while guaranteeing exponential stability for unperturbed systems and global uniform ultimate boundedness for bounded disturbances. Two implementations address computational complexity. The online version solves an optimization at each update for theoretical optimality. The offline version precomputes optimal horizons using conic partitioning, reducing online computation to a lookup. Simulations demonstrate 59-74\% reductions in sensor utilization compared to periodic sampling.

Three extensions merit investigation. First, adaptive triggering thresholds adjusting to observed disturbance patterns could improve robustness. Second, optimizing sensor subsets (multiple sensors per update) would provide finer control over communication-performance tradeoffs. Third, extending to nonlinear systems requires addressing the interaction between nonlinear dynamics and asynchronous state estimates. This framework enables resource-aware control in networked systems by optimizing sensor selection while maintaining stability guarantees.

\section*{Acknowledgments}
The author gratefully acknowledges Dr. Christophe Fiter for his significant contributions and dedicated supervision throughout this research conducted at CRIStAL in 2017--2018, and thanks Dr. Lotfi Belkoura for his valuable guidance and support.

\bibliographystyle{unsrtnat}
\bibliography{references} 

\section{Appendix: Proofs}\label{Appendix}
\begin{proof}[Proof of Proposition \ref{STMS-off-p}:]
	Consider $ x_k\in\mathbb{R}^n $	and the quadratic Lyapunov function $  V(x)=x^T P x $ in which the matrix $ P $ satisfies \eqref{LMI1-STM-on-p} and \eqref{LMI2-STM-on-p}.	
	\begin{enumerate}
		\item [Case (i):]  $ x_k \not\in \mathcal{E}(P,1)  $. Firstly we should prove that the set $ \bar{S}_{l_{min}}^{l_{max}}(\mathcal{S},x_k) $ in not empty (i.e. for a given $ x_k \in \in \mathbb{R}^n $, $\exists \sigma \in {\bar{S}}^{l_{max}}_{l_{min}}(\mathcal{S},x_k) $). From \eqref{LMI2-STM-on-p} and Schur complement, we have $ M=M^T\succeq 0 $, and $ \frac{\gamma}{\chi}I-PM^{-1}P- P \succeq 0 $. The latter equivalently can be written as 
		\begin{equation}\label{equival-LMI2}
			\frac{\gamma}{\chi} \geq \lambda_{max} \big(PM^{-1}P+P\big).   
		\end{equation}
		Consider a sampling horizon $ \sigma^* $ satisfying \eqref{LMI1-STM-on-p}. Then, from \eqref{equival-LMI2} and \eqref{U_sigma}, one has 	
		\begin{equation}
			\begin{pmatrix}
				x_k \\
				1
			\end{pmatrix}^T U_{\sigma^*} \begin{pmatrix}
				x_k \\
				1
			\end{pmatrix} \geq 0,
		\end{equation}	
		which implies that $ \sigma^* \in {\bar{S}}^{l_{max}}_{l_{min}}(\mathcal{S},x_k) $ (i.e. the set $\bar{S}_{l_{min}}^{l_{max}}(\mathcal{S},x_k) $ in not empty, for any $ x_k \notin \mathcal{E}(P,1)$).
		
		Now, consider $ \sigma_k \in \bar{S}_{l_{min}}^{l_{max}}(\mathcal{S},x_k) $. Then, we have	
		\begin{equation}\label{NinRP2-on}
			\begin{split}
				V(x_{k+1}) & = x_{k+1}^TPx_{k+1} \\
				& = \big(\Phi_{\sigma_k} x_k+\bar{w}_k\big)^TP\big(\Phi_{\sigma_k} x_k+\bar{w}_k\big) \\
				&= x^T_k \Phi^T_{{\sigma}_k} P \Phi_{{\sigma}_k} x_k+ x^T_k \Phi^T_{{\sigma}_k} P \bar{w}_k \\
				&\quad +\bar{w}^T_k P \Phi_{{\sigma}_k} x_k+\bar{w}^T_k P \bar{w}_k .
			\end{split}
		\end{equation}
		
		Using the inequality $ x^Ty+y^Tx \leq x^T M x+ y^T M^{-1} y $ and for any $ x\in \mathbb{R}^n $, $ y\in \mathbb{R}^n $, and $ M=M^T\succ 0 $, we have	
		\begin{equation}\label{inq}
		\begin{split}
				x^T_k \Phi^T_{{\sigma}_k} P \bar{w}_k+\bar{w}^T_k P \Phi_{{\sigma}_k} x_k & \leq x^T_k \Phi^T_{{\sigma}_k} M \Phi_{{\sigma}_k} x_k   \\
				&\quad + \bar{w}^T_k P M^{-1} P \bar{w}_k. 
		\end{split}
		\end{equation} 
		
		Therefore, from \eqref{NinRP2-on}, we have 
		\begin{equation}\label{NinRP2-onn}
			\begin{split}
				V(x_{k+1}) &= x^T_k \Phi^T_{{\sigma}_k} P \Phi_{{\sigma}_k} x_k+ x^T_k \Phi^T_{{\sigma}_k} P \bar{w}_k \\
				& \quad +\bar{w}^T_k P \Phi_{{\sigma}_k} x_k + \bar{w}^T_k P \bar{w}_k \\
				& \leq x^T_k \Phi^T_{{\sigma}_k}\Big( P + M \Big) \Phi_{{\sigma}_k} x_k   + \bar{w}^T_k \Big( P M^{-1} P + P \Big) \bar{w}_k ,
			\end{split}
		\end{equation}
		and furthermore
		\begin{equation}\label{NinRPLMI-on}
			\begin{split}
				V(x_{k+1}) \leq
				\begin{pmatrix}
					x_k\\
					1
				\end{pmatrix}^T \Lambda_{\sigma_k} \begin{pmatrix}
					x_k\\
					1
				\end{pmatrix},
			\end{split}
		\end{equation}
		in which
\begin{equation}\label{Lambda-0n}
	\Lambda_{\sigma_k} =
	\begin{multlined}[t]
		\begin{pmatrix}
			\Phi^T_{{\sigma}_k} \big(P+M\big)\Phi_{{\sigma}_k} & \textbf{0}\\[4ex]
			\textbf{0} & \Big(\varpi \sum_{q=0}^{|\sigma_k|-1}C^q \Big)^2 \\
			& \times \lambda_{max}(P M^{-1} P+P)
		\end{pmatrix}.
	\end{multlined}
\end{equation}
		
		Then, since $ x_k \notin \mathcal{E}(P,1) $ (i.e. $ x_k^T P x_k > 1 $), we have	
		\begin{equation}\label{NinRPLMION}
			\begin{split}
				V(x_{k+1}) & \leq
				\begin{pmatrix}
					x_k\\
					1
				\end{pmatrix}^T \Bigg[\Lambda_{\sigma_k}+ \begin{bmatrix}
					-\gamma P & \textbf{0}\\
					\textbf{0} & \gamma
				\end{bmatrix} \Bigg] \begin{pmatrix}
					x_k\\
					1
				\end{pmatrix},
			\end{split}
		\end{equation}	
		with $ \gamma>0 $. Then, from \eqref{es-on-p}, we get		
		\begin{equation}
			\begin{split}	
				V(x_{k+1}) & \leq \begin{pmatrix}
					x_k\\
					1
				\end{pmatrix}^T 
				\begin{pmatrix}
					e^{(-\beta\sum\nolimits_{j=1}^{|{\sigma_k} |}T_{\sigma_k}^{j }) }P-\gamma P & \textbf{0}\\
					\textbf{0} & \gamma
				\end{pmatrix}
				\begin{pmatrix}
					x_k\\
					1
				\end{pmatrix}\\
				& \leq e^{(-\beta\sum\nolimits_{j=1}^{|{\sigma_k} |}T_{\sigma_k}^{j }) } x_k^T P x_k + \gamma (1-x_k^T P x_k),
			\end{split}
		\end{equation}
		Therefore, since $x_k^T P x_k >  1 $, one has $ V(x_{k+1}) \leq e^{(-\beta|{\sigma_k} |T) } V(x_k)$. 	
		\item [ Case (ii):]  $ x_k \in \mathcal{E}(P,1) $. According to the self-triggering mechanism \eqref{STM-on-p} and from \eqref{CLSP}, we have	
		\begin{equation}
			x_{k+1} = \Phi_{\sigma_k}x_k+\bar{w}_k,
		\end{equation}
		where $ \sigma_k = \big(0\big) $. Then, $ \Phi_{\sigma_k}=\tilde{A}_{(0)} $ and $ \tau_{k+1}=\tau_k+T $. From the section \ref{sysref}, one has
		\begin{equation}\label{interball1}
			\begin{split}
				\Big\| x_{k+1} \Big\|_2 \leq \Big\| \Phi_{\sigma_k}x_k \Big\|_2+\Big\| \bar{w}_k \Big\|_2 & \leq C'\Big\|  x_k \Big\|_2 + \varpi \\
				& \leq \frac{C'}{\lambda_{min}(P)}+ \varpi=\sqrt{\eta}.
			\end{split}
		\end{equation}
		where $ C'=\Big\| \Phi_{\sigma_k}\Big\|_2=\Big\|\tilde{A}_{(0)} \Big\|_2 $  and $ \varpi $ is given in Assumption \ref{pertbounded}. We want to find a scalar $ \mu >0 $ such that  $  \mathcal{B}(0,\eta) \subset \mathcal{E}(P,\mu)  $. We know that $ x_k^T P x_k \leq \lambda_{max}(P) x_k^T x_k $, then form \eqref{interball1},
		
		\begin{equation}
			\mu \leq  \lambda_{max}(P) \Big(\frac{C'}{\lambda_{min}(P)}+ \varpi\Big)^2.
		\end{equation} 
		
	\end{enumerate}
	From  (i) and  (ii), one has $ x_k \in \mathcal{E}(P,\mu),~k\in\mathbb{N} $. Therefore, the solution of the system \eqref{CLSP} is globally uniformly ultimately bounded with the proposed self triggering mechanisms \eqref{STM-on-p}.
\end{proof}
\begin{proof}[Proof of Proposition \ref{STM-off-p}:]
	
	Consider $ x_k\in\mathbb{R}^n $	and the quadratic Lyapunov function $  V(x)=x^T P x $ in which the matrix $ P $ satisfies \eqref{SigmaOmega-off}.	
	\begin{enumerate}
		\item [Case (i):]  $ x_k \not\in \mathcal{E}(P,1)  $. Firstly we should prove that the set $ \bar{S}_{l_{min}}^{l_{max}}(\mathcal{S},x_k) $ in not empty (i.e. there exists at least the horizon $ \sigma^* \in {\bar{S}}^{l_{max}}_{l_{min}}(\mathcal{S},x_k),~ \forall x_k \in \mathbb{R}^n $). From \eqref{es-off-p}, the matrix $ U_c $ can be written as
		
		\begin{equation}
			U_c=U+\begin{pmatrix}
				\epsilon_c Q_c & * & *\\
				\textbf{0} & \textbf{0} & *\\
				\textbf{0} & \textbf{0} &  0 
			\end{pmatrix},
		\end{equation}
		which is composed of two semi-definite matrices. Then, clearly, it can be concluded that $ U_c $ is also a semi-definite matrix. Therefore, there exists at least the horizon $ \sigma^* \in {\bar{S}}^{l_{max}}_{l_{min}}(\mathcal{S},x_k),~ \forall x_k \in \mathbb{R}^n $. 
		
		Now, consider $ \sigma_k \in \bar{S}_{l_{min}}^{l_{max}}(\mathcal{S},x_k) $,		
		We want to guarantee that 
		\begin{equation}\label{NinRP}
			\begin{split}
				x_{k+1}^TPx_{k+1}<\bar{\beta} x_{k}^TPx_{k},~\bar{\beta}=e^{(-\beta|{\sigma} |T) }
			\end{split}
		\end{equation}
		
		From \eqref{CLSP} and \eqref{NinRP}, we have
\begin{equation}\label{NinRP2}
	\begin{aligned}
		&\big(\Phi_{\sigma_k} x_k+\bar{w}_k\big)^TP\big(\Phi_{\sigma_k} x_k+\bar{w}_k\big) \leq \bar{\beta} x_{k}^TPx_{k} \\
		&x^T_k \Phi^T_{{\sigma}_k} P \Phi_{{\sigma}_k} x_k+ x^T_k \Phi^T_{{\sigma}_k} P \bar{w}_k \\
		&\quad +\bar{w}^T_k P \Phi_{{\sigma}_k} x_k+\bar{w}^T_k P \bar{w}_k \\
		&\leq-\bar{\beta} x_{k}^TPx_{k},
	\end{aligned}
\end{equation}
		equivalently, equation \eqref{NinRP2} can be written as 
		\begin{equation}\label{NinRPLMI}
			\begin{split}
				\begin{pmatrix}
					x_k\\
					\bar{w}_k\\
					1
				\end{pmatrix}^T \Lambda_{\sigma_k} \begin{pmatrix}
					x_k\\
					\bar{w}_k\\
					1
				\end{pmatrix}>0,
			\end{split}
		\end{equation}
		in which
		\begin{equation}\label{Lambda}
			\begin{split}
				\Lambda_{\sigma_k} = \begin{pmatrix}
					-\Phi_{\sigma_k}^TP\Phi_{\sigma_k}+\bar{\beta}P & * & *\\
					-P\Phi_{\sigma_k} & -P & *\\
					\textbf{0} & \textbf{0} & 0
				\end{pmatrix}.
			\end{split}
		\end{equation}
		
		Since $ \Big\|\bar{w}_k\Big\|_2\leq  \varpi \sum_{q=0}^{|\sigma_k|-1}C^q=\theta $, this can be shown as 
		\begin{equation}\label{PertLMI}
			\begin{split}
				\begin{pmatrix}
					x_k\\
					\bar{w}_k\\
					1
				\end{pmatrix}^T      
				\begin{pmatrix}
					\textbf{0} & \textbf{0} & \textbf{0}\\
					\textbf{0} & \frac{-1}{ \theta}I & \textbf{0}\\
					\textbf{0} & \textbf{0} & 1
				\end{pmatrix}
				\begin{pmatrix}
					x_k\\
					\bar{w}_k\\
					1
				\end{pmatrix}\geq 0.
			\end{split}
		\end{equation}
		
		Also, the constraint $ x(k)\not\in \mathcal{E}(P,1) $ can be written as
		\begin{equation}\label{NILMI}
			\begin{split}
				\begin{pmatrix}
					x_k\\
					\bar{w}_k\\
					1
				\end{pmatrix}^T      
				\begin{pmatrix}
					P & \textbf{0} & \textbf{0}\\
					\textbf{0} & \textbf{0} & \textbf{0}\\
					\textbf{0} & \textbf{0} & -1
				\end{pmatrix}
				\begin{pmatrix}
					x_k\\
					\bar{w}_k\\
					1
				\end{pmatrix}\geq 0.
			\end{split}
		\end{equation}
		
		From equations \eqref{NinRPLMI},\eqref{PertLMI} and \eqref{NILMI} and by using S-procedure, the aim is to guarantee that
\begin{equation}\label{FLMI}
	\begin{multlined}[b]
		\begin{pmatrix}
			x_k\\
			\bar{w}_k\\
			1
		\end{pmatrix}^T \times \\
		\begin{pmatrix}
			-\Phi_{{\sigma}_k}^TP\Phi_{{\sigma}_k}+({\bar{\beta}}-{\gamma}_1) P & * & * \\
			-P\Phi_{{\sigma}_k} & \frac{{\gamma}_2}{ {\theta}}I-P & *\\
			\textbf{0} & \textbf{0} & {-{\gamma}_2+{\gamma}_1}
		\end{pmatrix} \\
		\times \begin{pmatrix}
			x_k\\
			\bar{w}_k\\
			1
		\end{pmatrix}\geq 0.
	\end{multlined}
\end{equation}
		in which $ \gamma_1, \gamma_2 > 0 $.
		
		Let $ x_{k} \in \mathbb{R}^n, k\in \mathbb{N} $, There exists a conic region $ \mathcal{R}_c, c \in \mathbb{N} $ as in \eqref{conicreg} such that $ x_{k} \in  \mathcal{R}_c$. Using the S-procedure, one can obtain that the inequality \eqref{FLMI} is satisfied if and only if there exists a scalar $ \epsilon_c >0 $ such that

\begin{equation}\label{LMI-FLMI}
	\begin{aligned}
		&\begin{pmatrix}
			\epsilon_c Q_c-\Phi_{{\sigma}_k}^TP\Phi_{{\sigma}_k}+({\bar{\beta}}-{\gamma}_1) P & \_ & \_ \\
			-P\Phi_{{\sigma}_k} & \frac{{\gamma}_2}{{\theta}}I-P & *\\
			\textbf{0} & \textbf{0} & {-{\gamma}_2+{\gamma}_1}
		\end{pmatrix} \\
		&\succeq 0
	\end{aligned}
\end{equation}

		Therefore, one has $ V(x_{k+1}) \leq e^{(-\beta|{\sigma_k} |T) } V(x_k)$ for every $ x_k \not\in \mathcal{E}(P,1)  $. 
		\item [Case (ii):]  $ x_k \in \mathcal{E}(P,1) $. According to the self-triggering mechanism \eqref{STM-on-p} and from \eqref{CLSP}, we have	
		\begin{equation}
			x_{k+1} = \Phi_{\sigma_k}x_k+\bar{w}_k,
		\end{equation}
		where $ \sigma_k = \big(0\big) $. Then, $ \Phi_{\sigma_k}=\tilde{A}_{(0)} $ and $ \tau_{k+1}=\tau_k+T $. From the section \ref{sysref}, one has
		\begin{equation}\label{interball}
			\begin{split}
				\Big\| x_{k+1} \Big\|_2 \leq \Big\| \Phi_{\sigma_k}x_k \Big\|_2+\Big\| \bar{w}_k \Big\|_2 & \leq C'\Big\|  x_k \Big\|_2 + \varpi \\
				& \leq \frac{C'}{\lambda_{min}(P)}+ \varpi=\sqrt{\eta}.
			\end{split}
		\end{equation}
		where $ C'=\Big\| \Phi_{\sigma_k}\Big\|_2=\Big\|\tilde{A}_{(0)} \Big\|_2 $  and $ \varpi $ is given in Assumption \ref{pertbounded}. We want to find a scalar $ \mu >0 $ such that  $  \mathcal{B}(0,\eta) \subset \mathcal{E}(P,\mu)  $. We know that $ x_k^T P x_k \leq \lambda_{max}(P) x_k^T x_k $, then form \eqref{interball},
		
		\begin{equation}
			\mu \leq  \lambda_{max}(P) \Big(\frac{C'}{\lambda_{min}(P)}+ \varpi\Big)^2.
		\end{equation}  
	\end{enumerate}
	From  (i) and  (ii), one has $ x_k \in \mathcal{E}(P,\mu),~k\in\mathbb{N} $. Therefore, the solution of the system \eqref{CLSP} is globally uniformly ultimately bounded with the proposed self triggering mechanisms \eqref{STMoff-p}.  
\end{proof}

\end{document}